\documentclass[journal]{IEEEtran}
\usepackage{amssymb, amsmath, color}
\usepackage{amsthm,graphicx,enumerate,verbatim,xcolor}
\usepackage[small]{caption}
\usepackage{bbm}
\usepackage{tikz}
\usetikzlibrary{arrows,positioning,fit,backgrounds}
\usetikzlibrary{calc}

\newtheorem{thm}{Theorem}

\newtheorem{lem}{Lemma}
\newtheorem{defn}{Definition}

\newtheorem{property}{Property}
\theoremstyle{remark}
\newtheorem{rem}{Remark}

\def\squarebox#1{\hbox to #1{\hfill\vbox to #1{\vfill}}}

\newcommand{\inout}{-}
\newcommand{\QEDA}{\hfill\ensuremath{\blacksquare}}

\newcommand{\Pbar}{\overline{P}}

\newcommand{\eps}{\varepsilon}

\newcommand{\Ebb}{\mathbb{E}}

\newcommand{\Pbb}{\mathbb{P}}
\newcommand{\Rbb}{\mathbb{R}}

\newcommand{\Bcal}{\mathcal{B}}
\newcommand{\Ccal}{\mathcal{C}}

\newcommand{\Kcal}{\mathcal{K}}

\newcommand{\Gcal}{\mathcal{G}}
\newcommand{\Ucal}{\mathcal{U}}
\newcommand{\Vcal}{\mathcal{V}}

\newcommand{\Xcal}{\mathcal{X}}
\newcommand{\Ycal}{\mathcal{Y}}

\newcommand{\Lcal}{\mathcal{L}}
\newcommand{\cn}{\mathcal{C}^n}
\newcommand{\Pbf}{\mathbf{P}}
\newcommand{\Qbf}{\mathbf{Q}}
\newcommand{\ls}{\limsup_{n\rightarrow\infty}}

\newcommand\blfootnote[1]{%
  \begingroup
  \renewcommand\thefootnote{}\footnote{#1}%
  \addtocounter{footnote}{-1}%
  \endgroup
}

\title{The Likelihood Encoder for Lossy Compression}

\author{ Eva~C.~Song, ~\IEEEmembership{Member,~IEEE,}
        Paul~Cuff, ~\IEEEmembership{Member,~IEEE,}
        and~ H.~Vincent~Poor,~\IEEEmembership{Fellow,~IEEE}
\thanks{The authors are with the Department
of Electrical Engineering, Princeton University, Princeton,
NJ 08544, USA  e-mail: \{csong, cuff, poor\}@princeton.edu.}
}

\begin{document}

\maketitle
\blfootnote{Parts of this work were presented at 2013 IEEE Information Theory Workshop (ITW) \cite{cuff-itw2013} and 2014 IEEE International Symposium on Information Theory (ISIT) \cite{song-isit14}.

Copyright (c) 2014 IEEE. Personal use of this material is permitted.  However, permission to use this material for any other purposes must be obtained from the IEEE by sending a request to pubs-permissions@ieee.org.}

\begin{abstract}
A likelihood encoder is studied in the context of lossy source compression. The analysis of the likelihood encoder is based on the soft-covering lemma. It is demonstrated that the use of a likelihood encoder together with the soft-covering lemma yields simple achievability proofs for classical source coding problems. The cases of the point-to-point rate-distortion function, the rate-distortion function with side information at the decoder (i.e. the Wyner-Ziv problem), and the multi-terminal source coding inner bound (i.e. the Berger-Tung problem) are examined in this paper. Furthermore, a non-asymptotic analysis is used for the point-to-point case to examine the upper bound on the excess distortion provided by this method. The likelihood encoder is also related to a recent alternative technique using properties of random binning.
\end{abstract}


\section{Introduction}
Rate-distortion theory, founded by Shannon in \cite{shannon-math} and \cite{shannon-rd}, provides the fundamental limits of lossy source compression. The minimum rate required to represent an independent and identically distributed (i.i.d.) source sequence under a given tolerance of distortion is given by the rate-distortion function. Related problems such as source coding with side information available at the decoder \cite{wz} and distributed source coding \cite{tung}, \cite{berger1977}, \cite{berger1989} have also been heavily studied in the past decades. Standard proofs \cite{cover}, \cite{network-it} of achievability for these rate-distortion problems often use joint-typicality encoding, i.e. the encoder looks for a codeword that is jointly typical with the source sequence. The distortion analysis involves bounding several ``error" events which may come from either encoding or decoding. These bounds use the joint asymptotic equipartition principle (J-AEP) and its immediate consequences as the main tool. In the cases where there are multiple information sources, such as side information at the decoder, intricacies arise, like the need for a Markov lemma \cite{cover}, \cite{network-it}. These subtleties also lead to error-prone proofs involving the analysis of error for random binning, which have been pointed out in several existing works \cite{hybrid}, \cite{lapidoth}.

In this work, we propose using a likelihood encoder to achieve these source coding results. The likelihood encoder is a stochastic encoder. For a chosen joint distribution $P_{XY}$, to encode a source sequence $x_1,...,x_n$ (i.e. $x^n$) with codebook $\Ccal^{(n)}=\{y^n(m)\}_m$, the encoder stochastically chooses an index $m$ with probability proportional to the likelihood of observing $x^n$  through the memoryless ``test channel" $P_{X|Y}$ given that the input is $y^n(m)$. 

The advantage of using such an encoder is that it naturally leads to an idealized distribution which is simple to analyze, based on the test channel. The distortion performance of the idealized distribution carries over to the source-reproduction joint distribution because the two distributions are shown to be close in total variation. Unlike the proof using the joint-typicality encoder, we do not need to identify different kinds of error events -- the distortion analysis of the idealized distribution is straightforward.

This proof technique of using an idealized distribution to approximate the source-reproduction joint distribution captures the performance of the encoder and decoder from a high level. Precise behaviors of the system are illuminated through the approximating distributions.  In other contexts, beyond the scope of this paper, this feature of the proof method can greatly simplify the analysis of secrecy and other objectives which demand comprehensive characterization of the behavior of the system.  In this paper we demonstrate this technique in more basic settings of rate-distortion theory, showing its effectiveness in simplifying and illuminating even those proofs.

Just as the joint-typicality encoder relies on the J-AEP, the likelihood encoder relies on the soft-covering lemma.\footnote{Note the difference between ``joint-typicality encoder" and the concept of ``joint-typicality". The concept of joint-typicality is used in the analysis of the soft-covering lemma, but the likelihood encoder itself is oblivious to this notion.} 
The idea of soft-covering was first introduced in \cite{wyner} and was later used in \cite{han-verdu} for channel resolvability. We introduced the idea of the likelihood encoder, in conjunction with the soft-covering lemma in \cite{cuff-isit2008} and \cite{cuff2012distributed} to achieve strong coordination and again used it in \cite{schieler-journal} for secrecy. Recent works in the literature have applied this tool in various other settings. In \cite{watanabe},  the soft-covering lemma and a smoothed version of the likelihood encoder are applied to derive one-shot achievability bounds for multiuser source coding problems. In \cite{slc}, the likelihood encoder is used in the proof for the Berger-Tung setting, although the analysis is quite different from the one used in this work. Similar ideas also arise in quantum information theory such as in \cite{datta}.

The application of the likelihood encoder together with the soft-covering lemma is not limited to only discrete alphabets. The proof for sources from continuous alphabets is readily included, since the soft-covering lemma imposes no restriction on alphabet size. Therefore, in contrast to \cite{network-it}, no extra work, i.e. quantization of the source, is needed to extend the standard proof for discrete sources to continuous sources. This advantage becomes more pronounced for the multi-terminal case, since generalization of the type-covering lemma and the Markov lemma to continuous alphabets is non-trivial. Although strong versions of the Markov lemma on finite alphabets that can prove the Berger-Tung inner bound can be found in \cite{network-it} and \cite{coord}, generalization to continuous alphabets is still an ongoing research topic. Some works, such as \cite{jeon} and \cite{mitran2010}, have been dedicated to making this transition, yet are not strong enough to be applied to the Berger-Tung case.

The rest of the paper is organized as follows. In Section \ref{prelim}, we introduce notation, some basic concepts and properties, define the likelihood encoder and give the soft-covering lemma. Sections \ref{p2p} to \ref{bt} deal with the point-to-point rate-distortion, Wyner-Ziv, and Berger-Tung problems, respectively, with increasing complexity. Within each of these sections, we first review the problem setup along with the result, and then give the achievability proof using the likelihood encoder. In Section \ref{non-asymptotic}, we apply a non-asymptotic analysis to the excess distortion for the point-to-point case. In Section \ref{binning}, we relate the likelihood encoder to a proportional-probability encoder \cite{rb-yassaee}, whose analysis is based on random-binning. Finally, in Section \ref{conclude}, we summarize the work.

\section{Preliminaries} \label{prelim}
\subsection{Notation}
A vector $(X_1,..., X_n)$ is denoted by $X^n$. Limits taken with respect to ``$n\rightarrow \infty$" are abbreviated as ``$\rightarrow_n$". When $X$ denotes a random variable, $x$ is used to denote a realization, $\mathcal{X}$ is used to denote the support of that random variable, and $\Delta_{\Xcal}$ is used to denote the probability simplex of distributions with alphabet $\Xcal$. A Markov relation is denoted by the symbol $\inout$. We use $\Ebb_P$, $\Pbb_P$, and $I_{P}(X;Y)$ to indicate expectation, probability, and mutual information taken with respect to a distribution $P$; however, when the distribution is clear from the context, the subscript will be omitted. To keep the notation uncluttered, the arguments of a distribution are sometimes omitted when the arguments' symbols match the subscripts of the distribution, e.g. $P_{X|Y}(x|y)=P_{X|Y}$. 
We use a bold capital letter $\mathbf{P}$ to denote that a distribution $P$ is random. We use $\Rbb$ to denote the set of real numbers and $\Rbb^+=[0,+\infty)$. 

For a per-letter distortion measure $d: \mathcal{X} \times \mathcal{Y}\mapsto \mathbb{R}^+$, we use $\Ebb [d(X,Y)]$ to measure the distortion of $X$ incurred by representing it as $Y$. The maximum distortion is defined as
\begin{eqnarray}
d_{max}=\max_{(x,y) \in \Xcal\times\Ycal}d(x,y).
\end{eqnarray}
The distortion between two sequences is defined to be the per-letter average distortion
\begin{eqnarray} 
d(x^n,y^n)=\frac1n\sum_{t=1}^n d(x_t,y_t).
\end{eqnarray}

\subsection{Total Variation Distance}
The total variation distance between two probability measures $P$ and $Q$ on the same $\sigma$-algebra $\mathcal{F}$ of subsets of the sample space $\Xcal$ is defined as
\begin{eqnarray}
\lVert P-Q\rVert_{TV}\triangleq \sup_{\mathcal{A}\in \mathcal{F}}|P(\mathcal{A})-Q(\mathcal{A})|.
\end{eqnarray}

\begin{property}[cf. \cite{schieler-journal}, Property 2] \label{property-tv}
Total variation distance satisfies the following properties:
\begin{enumerate}[(a)] 
\item \label{a} If $\Xcal$ is countable, then total variation can be rewritten as 
\begin{equation}
\lVert P - Q \rVert_{TV} = \frac12 \sum_{x\in\Xcal} |p(x)-q(x)|,
\end{equation}
where $p(\cdot)$ and $q(\cdot)$ are the probability mass functions of $X$ under $P$ and $Q$, respectively.
\item \label{b} Let $\eps>0$ and let $f(x)$ be a function in a bounded range with width $b \in\Rbb^+$. Then
\begin{equation} 
\label{tvcontinuous}
\lVert P-Q \rVert_{TV} < \eps \:\Longrightarrow\: \big| \Ebb_P[f(X)] - \Ebb_Q[f(X)] \big | < \eps b.
\end{equation}
\item \label{c} Total variation satisfies the triangle inequality. For any $P, Q, S \in \Delta_{\Xcal}$, 
\begin{equation}
\lVert P - Q \rVert_{TV} \leq \lVert P - S \rVert_{TV} + \lVert S - Q \rVert_{TV}.
\end{equation}
\item \label{d} Let $P_{X}P_{Y|X}$ and $Q_XP_{Y|X}$ be joint distributions on $\Delta_{\Xcal\times\Ycal}$. Then 
\begin{equation}
\lVert P_XP_{Y|X} - Q_X P_{Y|X} \rVert_{TV} = \lVert P_X - Q_X \rVert_{TV}.
\end{equation}
\item \label{e} For any $P,Q \in \Delta_{\Xcal\times\Ycal}$, 
\begin{equation}
\lVert P_X - Q_X \rVert_{TV} \leq \lVert P_{XY} - Q_{XY} \rVert_{TV},
\end{equation}
where $P_X$ and $Q_X$ are the marginals of $P$ and $Q$, respectively.
\end{enumerate}
\end{property}

\subsection{The Likelihood Encoder} \label{sub-le}
We now define the likelihood encoder, operating at rate $R$, which observes a sequence $x_1,...,x_n$ and maps it to a message $M \in \{1, \ldots ,2^{nR}\}$.  In normal usage, a decoder will then use $M$ to form an approximate reconstruction of the $x_1,...,x_n$ sequence.

The encoder is specified by a codebook $\{y^n(1), \ldots, y^n(2^{nR})\}$ and a joint distribution $P_{XY}$.  Consider the likelihood function for each codeword, with respect to a memoryless channel from $Y$ to $X$, defined as follows:
\begin{eqnarray}
\Lcal(m|x^n)&\triangleq& P_{X^n|Y^n}(x^n|y^n(m))\\
&\triangleq&\prod_{t=1}^nP_{X|Y}(x_t|y_t(m)).
\end{eqnarray}
A likelihood encoder is a stochastic encoder that determines the message index $m$ with probability proportional to $\Lcal(m|x^n)$, i.e. 
\begin{eqnarray}
P_{M|X^n}(m|x^n)=\frac{\Lcal(m|x^n)}{\sum_{m'\in\{1, \ldots, 2^{nR}\}}\Lcal(m'|x^n)}\propto \Lcal(m|x^n).
\end{eqnarray}

\subsection{Soft-Covering Lemma}
Now we introduce the core lemma that serves as the foundation for this analysis. One can consider the role of the soft-covering lemma in analyzing the likelihood encoder as analogous to that of the J-AEP which is used for the analysis of joint-typicality encoders. The general idea of the soft-covering lemma is that the distribution induced by selecting uniformly from a random codebook and passing the codeword through a memoryless channel is close to an i.i.d. distribution as long as the codebook size is large enough. The idea of soft covering originated from Wyner \cite{wyner} in the context of common information. Later, this result was generalized by \cite{han-verdu} to prove the fundamental limits of channel resolvability, which has both an achievability and a converse part, and has since been studied and strengthened in a variety of contexts. The version used in this paper is proved in \cite{cuff2012distributed} for channel synthesis.
\begin{lem}[\cite{cuff2012distributed}, Lemma IV.1]
Given a joint distribution $P_{XY}$, let $\Ccal^{(n)}$ be a random collection of sequences $Y^n(m)$, with $m=1,...,2^{nR}$, each drawn independently according to $\prod_{t=1}^n P_Y(y_t)$. Denote by $\Pbf_{X^n}$ the output distribution induced by independently selecting an index $M$ uniformly at random and applying $Y^n(M)$ to the memoryless channel specified by $P_{X|Y}$. Then if $R>I(X;Y)$,
\begin{eqnarray}
\Ebb_{\cn}\left[\left\lVert \Pbf_{X^n}-\prod_{t=1}^n P_X\right\rVert_{TV}\right]\rightarrow_n 0.
\end{eqnarray}
\end{lem}
The next three sections contain the focus of this paper, where we will use the soft-covering lemma to obtain simple achievability proofs for the rate-distortion function, the Wyner-Ziv problem, and the Berger-Tung inner bound for distributed source coding.

\section{The Point-to-Point Rate-Distortion Problem} \label{p2p}
Let us first start with point-to-point lossy compression. This simple setting outlines the key steps in the analysis, which are applied to the more complex settings in Section \ref{sec-wz} and \ref{bt}. 
\subsection{Problem Setup and Result Review} \label{sec-p2psetup}
Rate-distortion theory determines the optimal compression rate $R$ for an i.i.d. source sequence $X^n$ distributed according to $X_t \sim P_X$ with the following constraints:
\begin{itemize}
\item Encoder $f_n: \mathcal{X}^n \mapsto \mathcal{M}$ (possibly stochastic);
\item Decoder $g_n: \mathcal{M} \mapsto \mathcal{Y}^n$ (possibly stochastic);
\item Compression rate: $R$, i.e. $|\mathcal{M}|=2^{nR}$.
\end{itemize}
The system performance is measured according to the time-averaged distortion (as defined in the notation section):
\begin{itemize}
\item Average distortion: $d(X^n, Y^n)=\frac1n\sum_{t=1}^nd(X_t,Y_t)$.
\end{itemize}

\begin{defn} \label{df1}
A rate distortion pair $(R,D)$ is achievable if there exists a sequence of rate $R$ encoders and decoders $(f_n, g_n)$, such that 
\begin{eqnarray}
\ls\Ebb [d(X^n,Y^n)]\leq D.
\end{eqnarray}
\end{defn}
\begin{defn} \label{rate-distortion}
The rate distortion function is $R(D)\triangleq \inf_{\{(R,D) \text{ is achievable}\}} R$.
\end{defn}
The above mathematical formulation is illustrated in Fig. \ref{setup}.
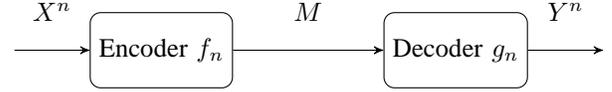
\begin{figure}
  \centering
\begin{tikzpicture}
[node distance=1cm,minimum height=10mm,minimum width=14mm,arw/.style={->,>=stealth'}]
  \node[coordinate] (source) {};
  \node[rectangle,draw,rounded corners] (encoder) [right =of source] {Encoder $f_n$};
  \node[rectangle,draw,rounded corners] (decoder) [right =2cm of encoder] {Decoder $g_n$};
  \node[coordinate] (sink)[right =of decoder] {};
    
  \draw [arw] (source) to node[midway,above]{$X^n$} (encoder);
  \draw [arw] (encoder) to node[midway,above]{$M$} (decoder);  
  \draw [arw] (decoder) to node[midway,above]{$Y^n$} (sink);
  
\end{tikzpicture}
\caption{Point-to-point lossy compression setup}
\label{setup}
\end{figure}
The characterization of this fundamental quantity in information theory was shown by Shannon \cite{shannon-math} (and e.g. \cite{cover} ) as
\begin{eqnarray}
R(D)&=& \min_{P_{Y|X}: \mathbb{E}[d(X,Y)]\leq D} I(X;Y), \label{rdfunction}
\end{eqnarray}
where the mutual information is taken with respect to $P_{XY}=P_X P_{Y|X}$. In other words, we are able to achieve distortion level $D$ with any rate greater than $R(D)$ given in $(\ref{rdfunction})$.

\subsection{Achievability Proof Using the Likelihood Encoder} \label{p2p-proof}
Here we make an additional note on the notation. In the following proof, $P$ is reserved for denoting the source-reproduction joint distribution, which we refer to as the system-induced distribution. Recall that bold letter $\Pbf$ indicates that the distribution itself is random because it is a function of the random codebook. The single letter distributions appearing in the right-hand side of $(\ref{rdfunction})$ are replaced with $\Pbar$ in the following proof to avoid confusion with the system-induced distribution. The marginal and conditional distributions derived from $\Pbar_{XY}$ are denoted as $\Pbar_X$, $\Pbar_Y$, $\Pbar_{X|Y}$ and $\Pbar_{Y|X}$. Note that the source $\Pbar_X=P_X$. We use $\Pbar_{X^nY^n}$ to denote an i.i.d. distribution, i.e. 
\begin{eqnarray}
\Pbar_{X^nY^n}=\prod_{t=1}^n\Pbar_{XY}, 
\end{eqnarray}
and similarly for the marginal and conditional distributions derived from $\Pbar_{XY}$.

\subsubsection{High-level outline}
To prove achievability, we will use the likelihood encoder and approximate the system-induced distribution by a well-behaved distribution. The soft-covering lemma allows us to claim that the approximating distribution matches the system.

Let $R > R(D)$, where $R(D)$ is from the right-hand side of $(\ref{rdfunction})$.  We prove that $R$ is achievable for distortion $D$. By the rate-distortion formula stated in $(\ref{rdfunction})$, we can fix $\Pbar_{Y|X}$ such that $R>I_{\Pbar}(X;Y)$ and $\Ebb_{\Pbar} [d(X,Y)]<D$. We will use the likelihood encoder derived from $\Pbar_{XY}$ and a random codebook $\{y^n(m)\}_m$ generated according to $\Pbar_{Y}$ to prove the result. The decoder will simply reproduce $y^n(M)$ upon receiving the message $M$. 

The joint distribution of source-index-reproduction induced by the encoder and decoder is
\begin{eqnarray}
&&\Pbf_{X^nMY^n}(x^n,m,y^n)\nonumber\\
&=&P_{X^n}(x^n)\Pbf_{M|X^n}(m|x^n)\Pbf_{Y^n|M}(y^n|m)\\
&\triangleq&P_{X^n}(x^n)\Pbf_{LE}(m|x^n)\Pbf_{D}(y^n|m) \label{sysind}
\end{eqnarray}
where $\Pbf_{LE}$ is the likelihood encoder and $\Pbf_{D}$ is a codeword lookup decoder. 

We will show that this is well approximated by the uniform distribution over the message index and a memoryless channel from the reconstruction sequence to the source sequence according to $\Pbar_{X|Y}$. This distribution will clearly achieve the desired distortion.

\subsubsection{Proof}
We now concisely restate the behavior of the encoder and decoder -- components of the induced distribution stated in $(\ref{sysind})$. These are derived from the distribution $\Pbar_{XY}$ stated in the outline.

{\textbf{Codebook generation}}: We independently generate $2^{nR}$ sequences in $\mathcal{Y}^n$ according to $\prod_{t=1}^n \Pbar_Y(y_t)$ and index them by $m\in\{1, \ldots,2^{nR}\}$. We use $\mathcal{C}^{(n)}$ to denote the random codebook.

{\textbf{Encoder}}: The encoder $\Pbf_{LE}(m|x^n)$ is the likelihood encoder that chooses $M$ stochastically with probability proportional to the likelihood function given by
\begin{eqnarray}
\Lcal(m|x^n)=\Pbar_{X^n|Y^n}(x^n|Y^n(m)).
\end{eqnarray}

{\textbf{Decoder}}: The decoder $\Pbf_{D}(y^n|m)$ is a codeword lookup decoder that simply reproduces $Y^n(m)$.

{\textbf{Analysis}}: We will consider two distributions for the analysis, the system-induced distribution $\Pbf$ and an approximating distribution $\Qbf$, which is much easier to analyze. We will show that $\Pbf$ and $\Qbf$ are close in total variation (on average over the random codebook). Hence, $\Pbf$ achieves the performance of $\Qbf$.

\begin{figure}
  \centering
\begin{tikzpicture}
[node distance=1cm,minimum height=10mm,minimum width=14mm,arw/.style={->,>=stealth'}]
  \node[coordinate] (source) {};
  \node[rectangle,draw,rounded corners] (encoder) [right =of source] {$\Ccal^{(n)}$};
  \node[rectangle,draw,rounded corners] (decoder) [right =2cm of encoder] {$\Pbar_{X|Y}$};
  \node[coordinate] (sink)[right =of decoder] {};
    
  \draw [arw] (source) to node[midway,above]{$M$} (encoder);
  \draw [arw] (encoder) to node[midway,above]{$Y^n(M)$} (decoder);  
  \draw [arw] (decoder) to node[midway,above]{$X^n$} (sink);
  
\end{tikzpicture}
\caption{Idealized distribution with test channel $\Pbar_{X|Y}$}
\label{test}
\end{figure}
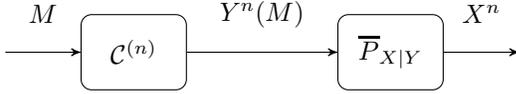

Design the approximating distribution $\Qbf$ via a uniform distribution over the same random codebook and a test channel $\Pbar_{X|Y}$ as shown in Fig. \ref{test}. We will refer to a distribution of this structure as an idealized distribution.
The joint distribution under the idealized distribution $\Qbf$ shown in Fig. \ref{test} can be written as 
\begin{eqnarray}
&&\mathbf{Q}_{X^nMY^n}(x^n,m,y^n)\nonumber\\
&=&Q_M(m)\mathbf{Q}_{Y^n|M}(y^n|m)\mathbf{Q}_{X^n|M}(x^n|m)\label{Qmarkov}\\
&=&\frac{1}{2^{nR}}\mathbbm{1}\{y^n=Y^n(m)\}\prod_{t=1}^n\Pbar_{X|Y}(x_t|Y_t(m))\\
&=&\frac{1}{2^{nR}}\mathbbm{1}\{y^n=Y^n(m)\}\prod_{t=1}^n\Pbar_{X|Y}(x_t|y_t).
\end{eqnarray}

The idealized distribution $\Qbf$ has the following property: for any $(x^n,y^n)\in \Xcal^n\times \Ycal^n$, 
\begin{eqnarray}
&&\mathbb{E}_{\mathcal{C}^{(n)}}[\mathbf{Q}_{X^nY^n}(x^n,y^n)]\nonumber\\
&=&\mathbb{E}_{\mathcal{C}^{(n)}}\left[\frac{1}{2^{nR}}\sum_{m}\mathbbm{1}\{y^n=Y^n(m)\}\right]\nonumber\\
&&\prod_{t=1}^n\Pbar_{X|Y}(x_t|y_t)\\
&=&\frac{1}{2^{nR}}\sum_{m}\mathbb{E}_{\mathcal{C}^{(n)}}[\mathbbm{1}\{y^n=Y^n(m)\}]\nonumber\\
&&\prod_{t=1}^n\Pbar_{X|Y}(x_t|y_t)\\
&=&\frac{1}{2^{nR}}\sum_{m}\Pbar_{Y^n}(y^n)\prod_{t=1}^n\Pbar_{X|Y}(x_t|y_t)\\
&=&\Pbar_{X^nY^n}(x^n,y^n) \label{expectation}
\end{eqnarray}
where $\Pbar_{X^nY^n}$ denotes the i.i.d. distribution $\prod_{t=1}^n\Pbar_{XY}$. This implies, in particular, that the distortion under the idealized distribution $\Qbf$ averaged over the random codebook, conveniently simplifies to $\Ebb_{\Pbar} [d(X,Y)]$. That is,
\begin{eqnarray}
&&\mathbb{E}_{\mathcal{C}^{(n)}}\left[ \mathbb{E}_\mathbf{Q}[d(X^n,Y^n)]\right]\nonumber\\
&=&\mathbb{E}_{\mathcal{C}^{(n)}}\left[\sum_{x^n,y^n}\mathbf{Q}(x^n,y^n)d(x^n,y^n)\right]\label{n4}\\
&=&\sum_{x^n,y^n}\mathbb{E}_{\mathcal{C}^{(n)}}[\mathbf{Q}(x^n,y^n)] d(x^n,y^n)\\
&=&\sum_{x^n,y^n}\Pbar_{X^n,Y^n}(x^n,y^n)d(x^n,y^n) \label{takingexp}\\
&=&\mathbb{E}_{\Pbar}[d(X^n,Y^n)]\\
&=&\Ebb_{\Pbar} [d(X,Y)], \label{distortion-iid}
\end{eqnarray}
where $(\ref{takingexp})$ follows from $(\ref{expectation})$. It is worth emphasizing that although $\Qbf_{X^nY^n}$ is very different from the i.i.d. distribution $\Pbar_{X^nY^n}$, it is exactly the i.i.d. distribution when averaged over codebooks and thus achieves the same expected distortion.

Below is the key observation given in $(\ref{encp2p})$ and $(\ref{decp2p})$: the conditional distributions under $\Qbf$ match our choice of encoder and decoder under the system-induced distribution $\Pbf$. In fact, our motivation for using the likelihood encoder comes from this construction of $\Qbf$. Notice that
\begin{eqnarray}
\mathbf{Q}_{M|X^n}(m|x^n)=\Pbf_{LE}(m|x^n), \label{encp2p}
\end{eqnarray}
and
\begin{eqnarray}
\mathbf{Q}_{Y^n|M}(y^n|m)=\Pbf_{D}(y^n|m). \label{decp2p}
\end{eqnarray}

Now invoking the soft-covering lemma, since $R>I_{\Pbar}(X;Y)$, we have
\begin{eqnarray}
\mathbb{E}_{\mathcal{C}^{(n)}}\left[\lVert\Pbar_{X^n}-\mathbf{Q}_{X^n}\rVert_{TV}\right]\leq \epsilon_n,
\end{eqnarray}
where $\epsilon_n\rightarrow_n 0$. This gives us 
\begin{eqnarray}
&&\mathbb{E}_{\mathcal{C}^{(n)}}\left[\lVert\mathbf{P}_{X^nY^n}-\mathbf{Q}_{X^nY^n}\rVert_{TV}\right]\nonumber\\
&\leq&\mathbb{E}_{\mathcal{C}^{(n)}}\left[\lVert\mathbf{P}_{X^nY^nM}-\mathbf{Q}_{X^nY^nM}\rVert_{TV}\right]\label{n9}\\
&\leq&\epsilon_n,\label{bound-tv}
\end{eqnarray}
where $(\ref{n9})$ follows from Property \ref{property-tv}$(\ref{e})$ and $(\ref{bound-tv})$ follows from $(\ref{encp2p})$, $(\ref{decp2p})$ and Property \ref{property-tv}$(\ref{d})$. 

By Property \ref{property-tv}$(\ref{b})$,
\begin{eqnarray}
&&\left\vert\Ebb_{\mathbf{P}}[d(X^n,Y^n)]-\Ebb_{\mathbf{Q}}[d(X^n,Y^n)]\right|\nonumber\\
&\leq& d_{max}\lVert\mathbf{P}-\mathbf{Q}\rVert_{TV}.\label{expected-distortion}
\end{eqnarray}

Now we apply the random coding argument.
\begin{eqnarray}
&&\Ebb_{\Ccal^{(n)}}\left[\Ebb_{\mathbf{P}}[d(X^n,Y^n)]\right]\nonumber\\
&\leq&\Ebb_{\Ccal^{(n)}}\left[\Ebb_{\Qbf} [d(X^n,Y^n)]\right]\nonumber\\
&&+\Ebb_{\Ccal^{(n)}}\left[\left|\Ebb_{\mathbf{P}}[d(X^n,Y^n)]-\Ebb_{\mathbf{Q}}[d(X^n,Y^n)]\right|\right]\label{mm1}\\
&\leq&\Ebb_{\Pbar}[d(X,Y)]\nonumber\\
&&+d_{max}\Ebb_{\Ccal^{(n)}}\left[\lVert\mathbf{P}_{X^nY^n}-\mathbf{Q}_{X^nY^n}\rVert_{TV}\right]\label{n2}\\
&\leq&\Ebb_{\Pbar}[d(X,Y)]+d_{max}\epsilon_n \label{n3}
\end{eqnarray}
where $(\ref{n2})$ follows from $(\ref{distortion-iid})$ and $(\ref{expected-distortion})$; $(\ref{n3})$ follows from $(\ref{bound-tv})$.
Taking the limit on both sides of the inequalities gives us
\begin{eqnarray}
\ls\Ebb_{\Ccal^{(n)}}\left[\Ebb_{\mathbf{P}}[d(X^n,Y^n)]\right] \leq D.
\end{eqnarray}
Therefore, there exists a codebook satisfying the requirement. \QEDA

\begin{rem}
As the proof emphasizes, the distribution $\mathbf{Q}$ serves as an accurate approximation to the true system behavior, and this is not unique to the likelihood encoder.  In \cite{schieler-itw2013} a converse statement is shown.  That is, any efficient source encoding satisfying a distortion constraint behaves like $\mathbf{Q}$ as measured by normalized divergence.  However, a stochastic encoder is generally required for the approximation to hold in total variation.  Furthermore, for the likelihood encoder, the accuracy of this approximation is easily verified using the soft-covering lemma.  For other encoders, the proof of the fact that $\Qbf$ is a good approximation to the induced $\Pbf$ requires more effort to establish.
\end{rem}

\subsection{Excess Distortion} \label{sec-excess}
The proof above is for the average distortion criterion, i.e.
\begin{eqnarray}
\ls\Ebb\left[\sum_{t=1}^nd(X_t,Y_t)\right]\leq D. 
\end{eqnarray}
However, it is not hard to modify the proof to show that it also holds for excess distortion. 

With the same setup as in Section \ref{sec-p2psetup}, we change the average distortion requirement in the definition of achievability (Definition \ref{df1}) to the requirement that
\begin{eqnarray}
\Pbb\left[d(X^n,Y^n)>D \right]\rightarrow_n 0.
\end{eqnarray}
The corresponding rate-distortion function is still given by $R(D)$ in $(\ref{rdfunction})$. 

For the excess distortion, we use the exact same encoding/decoding scheme, along with the same random codebook $\mathcal{C}^n$, from Section \ref{p2p-proof}. We make the following modifications.

We replace $(\ref{n4})$ to $(\ref{distortion-iid})$ with
\begin{eqnarray}
&&\mathbb{E}_{\mathcal{C}^{(n)}}\left[ \mathbb{P}_\mathbf{Q} \left[d(X^n,Y^n)>D \right]\right]\nonumber\\
&=&\mathbb{E}_{\mathcal{C}^{(n)}}\left[\sum_{x^n,y^n}\mathbf{Q}(x^n,y^n)\mathbbm{1}\{d(X^n,Y^n)>D\}\right]\\
&=&\sum_{x^n,y^n}\mathbb{E}_{\mathcal{C}^{(n)}}[\mathbf{Q}(x^n,y^n)] \mathbbm{1}\{d(x^n,y^n)>D\}\\
&=&\sum_{x^n,y^n}\Pbar_{X^n,Y^n}(x^n,y^n)\mathbbm{1}\{d(x^n,y^n)>D\} \\
&=&\mathbb{P}_{\Pbar}[d(X^n,Y^n)>D], \label{mm2}
\end{eqnarray}
and replace $(\ref{mm1})$ to $(\ref{n3})$ with
\begin{eqnarray}
&&\mathbb{E}_{\mathcal{C}^{(n)}}\left[\Pbb_{\mathbf{P}}[d(X^n,Y^n)>D]\right]\nonumber\\
&\leq&\mathbb{E}_{\mathcal{C}^{(n)}}\left[\Pbb_{\Qbf}[d(X^n,Y^n)>D]\right]+\epsilon_n\\
&=&\Pbb_{\Pbar}\left[d(X^n,Y^n)>D\right]+\epsilon_n \label{ww1}
\end{eqnarray}
where the last step follows from $(\ref{mm2})$. Therefore, there exists a codebook that satisfies the requirement. \QEDA

\section{The Wyner-Ziv Problem} \label{sec-wz}
In this section, we use the mechanism that was established in Section \ref{p2p} and build upon it to solve a more complicated problem. The Wyner-Ziv problem, that is, the rate-distortion function with side information at the decoder, was solved in \cite{wz}. 

\subsection{Problem Setup and Result Review}
The source and side information pair $(X^n,Z^n)$ is distributed i.i.d. according to $(X_t,Z_t) \sim P_{XZ}$. The system has the following constraints:
\begin{itemize}
\item Encoder $f_n: \mathcal{X}^n \mapsto \mathcal{M}$ (possibly stochastic);
\item Decoder $g_n: \mathcal{M}\times \mathcal{Z}^n \mapsto \mathcal{Y}^n$ (possibly stochastic);
\item Compression rate: $R$, i.e. $|\mathcal{M}|=2^{nR}$.
\end{itemize}
The system performance is measured according to the time-averaged distortion (as defined in the notation section):
\begin{itemize}
\item Average distortion: $d(X^n, Y^n)=\frac1n\sum_{t=1}^n d(X_t,Y_t)$.\\
\end{itemize}

\begin{defn}
A rate distortion pair $(R,D)$ is achievable if there exists a sequence of rate $R$ encoders and decoders $(f_n, g_n)$, such that 
\begin{eqnarray}
\ls\Ebb \left[d(X^n,Y^n)\right]\leq D.
\end{eqnarray}
\end{defn}
\begin{defn} \label{rate-distortion-wz}
The rate distortion function is $R(D)\triangleq \inf_{\{(R,D) \text{ is achievable}\}} R$.
\end{defn}

The above mathematical formulation is illustrated in Fig. \ref{setup_wz}.
\begin{figure}
  \centering
\begin{tikzpicture}
[node distance=1cm,minimum height=10mm,minimum width=14mm,arw/.style={->,>=stealth'}]
  \node[coordinate] (source) {};
  \node[rectangle,draw,rounded corners] (encoder) [right =of source] {Encoder $f_n$};
  \node[rectangle,draw,rounded corners] (decoder) [right =2cm of encoder] {Decoder $g_n$};
  \node[coordinate] (sink)[right =of decoder] {};
  \node[coordinate] (side)[above =0.5cm of decoder] {};
    
  \draw [arw] (source) to node[midway,above]{$X^n$} (encoder);
  \draw [arw] (encoder) to node[midway,above]{$M$} (decoder);  
  \draw [arw] (decoder) to node[midway,above]{$Y^n$} (sink);
  \draw [arw] (side) to node[midway,above]{$Z^n$} (decoder);
  
\end{tikzpicture}
\caption{Rate-distortion theory for source coding with side information at the decoder -- the Wyner-Ziv problem}
\label{setup_wz}

\end{figure}
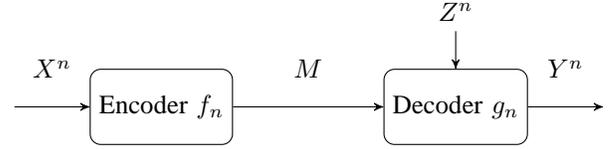

As mentioned previously, the solution to this source coding problem is given in \cite{wz}. The rate-distortion function with side information at the decoder is 
\begin{eqnarray}
R(D)&=& \min_{P_{V|XZ}\in \mathcal{M}(D)} I(X;V|Z), \label{rate}
\end{eqnarray}
where 
\begin{eqnarray}
\mathcal{M}(D)&=&\bigg\{P_{V|XZ}: V\inout X\inout Z, \nonumber\\
&&|\Vcal|\leq|\Xcal|+1, \nonumber\\
&&\text{and there exists }\text{a function } \phi \text{ s.t. }\nonumber\\
&&\mathbb{E}\left[d(X,Y)\right]\leq D ,Y\triangleq \phi(V,Z) \bigg\}. \label{md}
\end{eqnarray}

\subsection{Achievability Proof Using the Likelihood Encoder}
Before going into the main proof, let us first establish a property of total variation that will be helpful for both the Wyner-Ziv problem and the Berger-Tung inner bound.
\begin{lem} \label{helper}
For a distribution $P_{UVX}\in\Delta_{\Ucal\times\Ucal\times \Xcal}$ and $0<\eps<1$, if $\mathbb{P}[U\neq V]\leq \eps$, we have 
\begin{eqnarray}
\lVert P_{UX}-P_{VX}\rVert_{TV}\leq \eps.
\end{eqnarray}
\end{lem}
\begin{proof}
By definition,
\begin{eqnarray}
&&\lVert P_{UX}-P_{VX}\rVert_{TV}\nonumber\\
&=&\sup_{\mathcal{A\in\mathcal{F}}}\left\vert\mathbb{P}[(U,X)\in \mathcal{A}]-\mathbb{P}[(V,X)\in \mathcal{A}]\right\vert,
\end{eqnarray}
where $\mathcal{F}$ is the sigma-algebra on which $P_{UX}$ and $P_{VX}$ are defined and $\mathcal{A}$ represents a subset on the sample space $\Ucal\times\Xcal$.

Since for every $\mathcal{A} \in \mathcal{F}$
\begin{eqnarray}
&&\left\vert\mathbb{P}[(U,X)\in \mathcal{A}]-\mathbb{P}[(V,X)\in \mathcal{A}]\right\vert \nonumber\\
&\leq& \mathbb{P}[(U,X)\in \mathcal{A}]-\mathbb{P}[(V,X)\in \mathcal{A}, (U,X)\in  \mathcal{A}]\\
&=&\mathbb{P}[(U,X)\in \mathcal{A}, (V,X)\notin  \mathcal{A}]\\
&\leq&\mathbb{P}[U\neq V]\\
&\leq&\epsilon,
\end{eqnarray}
we have
\begin{eqnarray}
\sup_{\mathcal{A\in\mathcal{F}}}\left\vert\mathbb{P}[(U,X)\in \mathcal{A}]-\mathbb{P}[(V,X)\in \mathcal{A}]\right\vert\leq \epsilon.
\end{eqnarray}
\end{proof}

Here again to be consistent, we reserve $P$ for the system-induced distribution, with bold $\Pbf$ indicating that the distribution itself is random with respect to the random codebook. We replace the single-letter distributions appearing in the right-hand side of $(\ref{rate})$ and $(\ref{md})$ with $\Pbar$ and any marginal or conditional distributions derived from the joint single-letter distribution $\Pbar_{XZV}$. Note that $\Pbar_{XZ}=P_{XZ}$. We use $\Pbar_{X^nZ^nV^n}$ to denote an i.i.d. distribution, i.e.
\begin{eqnarray}
\Pbar_{X^nZ^nV^n}=\prod_{t=1}^n\Pbar_{XZV}.
\end{eqnarray}

\subsubsection{High-level outline}
We are now ready to give the achievability proof of $(\ref{rate})$. We introduce a virtual message which is produced by the encoder but not physically transmitted to the receiver so that this virtual message together with the actual message gives a high enough rate for applying the soft-covering lemma. Then we show that this virtual message can be reconstructed with vanishing error probability at the decoder by using the side information. This is analogous to the technique of random binning, where the index of the codeword within the bin is equivalent to the virtual message in our method.

Our proof technique again involves showing that the behavior of the system is approximated by a well-behaved distribution. The soft-covering lemma and channel decoding error bounds are used to analyze how well the approximating distribution matches the system.

Let $R>R(D)$ claimed in (\ref{rate}). We prove that $R$ is achievable for distortion $D$. Let $M'$ be a virtual message with rate $R'$ which is not physically transmitted. By the rate-distortion formula in $(\ref{rate})$, we can fix $R'$ and $\Pbar_{V|XZ}\in\mathcal{M}(D)$ ($\Pbar_{V|XZ}=\Pbar_{V|X}$) such that $R+R'>I_{\Pbar}(X;V)$ and $R'<I_{\Pbar}(V;Z)$, and there exists a function $\phi(\cdot,\cdot)$ yielding $Y= \phi(V,Z)$ and $\mathbb{E}\left[d(X,Y)\right]\leq D$. We will use the likelihood encoder derived from $\Pbar_{XV}$ and a random codebook $\{v^n(m,m')\}$ generated according to $\Pbar_V$ to prove the result. The decoder will first use the transmitted message $M$ and the side information $Z^n$ to decode $M'$ as $\hat{M}'$ and reproduce $v^n(M,\hat{M}')$. Then the reconstruction $Y^n$ is produced as a symbol-by-symbol application of $\phi(\cdot,\cdot)$ to $Z^n$ and $V^n$.

The distribution induced by the source, side information, encoder and decoder is
\begin{eqnarray}
&&\Pbf_{X^nZ^nMM'\hat{M}'Y^n}(x^n,z^n,m,m',\hat{m}',y^n)\nonumber\\
&=& P_{X^nZ^n}(x^n,z^n)\Pbf_{MM'|X^n}(m,m'|x^n) \nonumber\\
&&\Pbf_{\hat{M}'|MZ^n}(\hat{m}'|m,z^n)\Pbf_{Y^n|M\hat{M}'Z^n}(y^n|m,\hat{m}',z^n) \label{jointPP}\\
&\triangleq& P_{X^nZ^n}(x^n,z^n) \Pbf_{LE}(m,m'|x^n)\nonumber\\
&& \Pbf_D(\hat{m}'|m,z^n) \Pbf_\Phi(y^n|m,\hat{m}',z^n), \label{jointPP2}
\end{eqnarray}
where $\Pbf_{LE}(m,m'|x^n)$ is the likelihood encoder; $\Pbf_D(\hat{m}'|m,z^n)$ is the first part of the decoder that decodes $m'$ as $\hat{m}'$; and $\Pbf_\Phi(y^n|m,\hat{m}',z^n)$ is the second part of the decoder that reconstructs the source sequence. 

Two approximating distributions are used. The first is a uniform distribution over the pairs of messages $(m,m')$ with a memoryless channel from the associated codeword to the source sequence according to $\Pbar_{XZ|V}$. The second allows the decoder to form a reconstruction based on the $m'$ selected by the encoder rather than its own estimate of it.
\subsubsection{Proof}
We now concisely restate the behavior of the encoder and decoder, as these components of the system-induced distribution.

{\textbf{Codebook generation}}: We independently generate $2^{n(R+R')}$ sequences in $\mathcal{V}^n$ according to $\prod_{i=1}^n \Pbar_V(v_i)$ and index by $(m,m')\in\{1, \ldots ,2^{nR}\}\times\{1, \ldots ,2^{nR'}\}$. We use $\Ccal^{(n)}$ to denote the random codebook.

{\textbf{Encoder}}: The encoder $\Pbf_{LE}(m,m'|x^n)$ is the likelihood encoder that chooses $M$ and $M'$ stochastically with probability proportional to the likelihood function given by 
\begin{eqnarray}
\Lcal(m,m'|x^n)=\Pbar_{X^n|V^n}(x^n|V^n(m,m')).
\end{eqnarray}
Then the encoder sends $M$.

{\textbf{Decoder}}: The decoder receives $M$ and the side information $Z^n$ and has two decoding steps. In the first step, the decoder reconstructs $M'$ as $\hat{M}'$: let $\Pbf_D(\hat{m}'|m,z^n)$ be a good channel decoder (e.g. the maximum likelihood decoder) with respect to the sub-codebook ${\Ccal^{(n)}}(m)=\{v^n(m,a)\}_a$ and the memoryless channel $\Pbar_{Z|V}$. In the second step, the decoder forms the reconstruction of the source: let $\phi(\cdot,\cdot)$ be the function corresponding with the choice of $\Pbar_{V|XZ}$ in $(\ref{md})$; that is, $Y=\phi(V,Z)$ and $\Ebb_{\Pbar}\left[d(X,Y)\right]\leq D$. Define $\phi^n(v^n, z^n)$ as the concatenation $\{\phi(v_t,z_t)\}_{t=1}^n$ and set the decoder $\Pbf_\Phi$ to be the deterministic function
\begin{eqnarray}
\Pbf_\Phi(y^n|m,\hat{m}',z^n)\triangleq \mathbbm{1}\{y^n=\phi^n(V^n(m,\hat{m}'),z^n)\}.
\end{eqnarray}

{\textbf{Analysis:}} We consider three distributions for the analysis, the induced distribution $\Pbf$ and two approximating distributions $\Qbf^{(1)}$ and $\Qbf^{(2)}$. The idea is to show that 1) the system behaves well under $\Qbf^{(2)}$; and 2) $\Pbf$ and $\Qbf^{(2)}$ are close in total variation (on average over the random codebook) through $\Qbf^{(1)}$. 

The first approximating distribution, $\Qbf^{(1)}$, changes the distribution induced by the likelihood encoder to a distribution based on a reverse memoryless channel, as in the proof of point-to-point rate-distortion theory, and shown in Fig. \ref{auxiliary}. This is shown to be a good approximation using the soft-covering lemma. The second approximating distribution, $\Qbf^{(2)}$, pretends that $M'$, the index which is not transmitted, is used by the decoder to form the reconstruction. This is a good approximation because the decoder can accurately estimate $M'$.
\begin{figure}
  \centering
\begin{tikzpicture}
[node distance=1cm,minimum height=12mm,minimum width=14mm,arw/.style={->,>=stealth'}]
  \node[coordinate] (mm)[yshift=4mm] {};
  \node[coordinate] (mp)[yshift=-4mm] {};
  \node[rectangle,draw,rounded corners] (encoder) [right =of source] {$\Ccal^{(n)}$};
  \node[rectangle,draw,rounded corners] (decoder) [right =3cm of encoder] {$\Pbar_{XZ|V}$};
  \node[coordinate](xn) [right =of decoder, yshift=4mm] {};
  \node[coordinate](bn) [right =of decoder, yshift=-4mm] {};
    
  \draw [arw] (mm) to node[above, yshift=-1mm]{$M$} (mm -| encoder.west);
  \draw [arw] (mp) to node[below, yshift=1mm]{$M'$} (mp -| encoder.west);
  \draw [arw] (encoder) to node[midway,above]{$V^n(M,M')$} (decoder);  
  \draw [arw] (decoder.east |- xn) to node[above,yshift=-1mm]{$X^n$} (xn);
  \draw [arw] (decoder.east |- bn) to node[below,yshift=1mm]{$Z^n$} (bn);
\end{tikzpicture}
\caption{Idealized distribution with test channel $\Pbar_{XZ|V}$}
\label{auxiliary}
\end{figure}
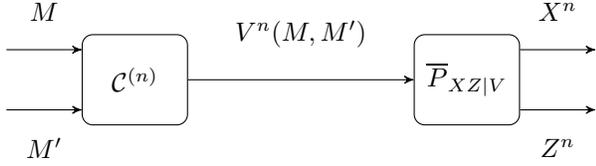

Both approximating distributions $\Qbf^{(1)}$ and $\Qbf^{(2)}$ are built upon the idealized marginal distribution over the information sources and messages, according to the test channel, as shown in Fig. \ref{auxiliary}. Note that this idealized distribution $\Qbf$ is no different from the one we considered for the point-to-point case, except for there being two message indices and two channel outputs. 
The joint distribution under $\Qbf$ in Fig. \ref{auxiliary} can be written as
\begin{eqnarray}
&&\mathbf{Q}_{X^nZ^nV^nMM'}(x^n,z^n,v^n,m,m')\nonumber\\
&=&Q_{MM'}(m,m')\Qbf_{V^n|MM'}(v^n|m,m')\nonumber\\
&&\mathbf{Q}_{X^nZ^n|MM'}(x^n,z^n|m,m')\label{n5}\\
&=&\frac{1}{2^{n(R+R')}}\mathbbm{1}\{v^n=V^n(m,m')\}\nonumber\\
&&\prod_{t=1}^n\Pbar_{XZ|V}(x_t,z_t|V_t(m,m'))\\
&=&\frac{1}{2^{n(R+R')}}\mathbbm{1}\{v^n=V^n(m,m')\}\nonumber\\
&&\prod_{t=1}^n\Pbar_{X|V}(x_t|v_t)\Pbar_{Z|X}(z_t|x_t), \label{markovchain}
\end{eqnarray}
where $(\ref{markovchain})$ follows from the Markov chain under $\Pbar$, $V\inout X\inout Z$. Note that by using the likelihood encoder, the idealized distribution $\Qbf$ satisfies
\begin{eqnarray}
\mathbf{Q}_{MM'|X^nZ^n}(m,m'|x^n,z^n)= \mathbf{P}_{LE}(m,m'|x^n). \label{enc}
\end{eqnarray}
Furthermore, using the same technique as $(\ref{expectation})$ and $(\ref{distortion-iid})$ given in the previous section, it can be verified that
\begin{eqnarray}
\mathbb{E}_{\mathcal{C}^{(n)}}\left[\mathbf{Q}_{X^nZ^nV^n}(x^n,z^n,v^n)\right]=\Pbar_{X^nZ^nV^n}(x^n,z^n,v^n). \label{expectationQ}
\end{eqnarray}
Consequently, 
\begin{eqnarray}
&&\mathbb{E}_{\mathcal{C}^{(n)}}\left[\Ebb_{\Qbf}\left[d\left(X^n,\phi^n(V^n,Z^n)\right)\right]\right]\nonumber\\
&=&\Ebb_{\Pbar}\left[d\left(X^n,\phi^n(V^n,Z^n)\right)\right].
\end{eqnarray}

Define the two distributions $\Qbf^{(1)}$ and $\Qbf^{(2)}$ based on $\Qbf$ as follows:
\begin{eqnarray}
&&\mathbf{Q}^{(1)}_{X^nZ^nMM'\hat{M}'Y^n}(x^n,z^n,m,m',\hat{m}',y^n)\nonumber\\
&\triangleq&\mathbf{Q}_{X^nZ^nMM'}(x^n,z^n,m,m') \Pbf_D(\hat{m}'|m,z^n)\nonumber\\
&&\Pbf_\Phi(y^n|m,\hat{m}',z^n) \label{Q1}\\
\nonumber\\
&&\mathbf{Q}^{(2)}_{X^nZ^nMM'\hat{M}'Y^n}(x^n,z^n,m,m',\hat{m}',y^n)\nonumber\\
&\triangleq&\mathbf{Q}_{X^nZ^nMM'}(x^n,z^n,m,m') \Pbf_D (\hat{m}'|m,z^n)\nonumber\\
&&\Pbf_\Phi(y^n|m,m',z^n). \label{Q2}
\end{eqnarray}
Notice that $\Qbf^{(2)}$ differs from $\Qbf^{(1)}$ by allowing the decoder to use $m'$ rather than $\hat{m}'$ when forming its reconstruction through $\phi^n$.

Therefore, on account of $(\ref{expectationQ})$, 
\begin{eqnarray}
\mathbb{E}_{\mathcal{C}^{(n)}} \left[\mathbf{Q}^{(2)}_{X^nZ^nY^n}(x^n,z^n,y^n)\right]=\Pbar_{X^nZ^nY^n}(x^n,z^n,y^n).
\end{eqnarray}

Now applying the soft-covering lemma, since $R+R'>I_{\Pbar}(X;V)=I_{\Pbar}(Z,X;V)$, we have
\begin{eqnarray}
\mathbb{E}_{\mathcal{C}^{(n)}}\left[\lVert\Pbar_{X^nZ^n}-\mathbf{Q}_{X^nZ^n}\rVert_{TV}\right]\leq \epsilon_n\rightarrow_n 0.
\end{eqnarray}
And with $(\ref{jointPP2})$, $(\ref{enc})$, $(\ref{Q1})$ and Property \ref{property-tv}$(\ref{d})$, we obtain
\begin{eqnarray}
&&\mathbb{E}_{\mathcal{C}^{(n)}}\left[\lVert \mathbf{P}_{X^nZ^nMM'\hat{M}'Y^n}-\mathbf{Q}^{(1)}_{X^nZ^nMM'\hat{M}'Y^n}\rVert_{TV}\right]\nonumber\\
&=&\mathbb{E}_{\mathcal{C}^{(n)}}\left[\lVert\Pbar_{X^nZ^n}-\mathbf{Q}_{X^nZ^n}\rVert_{TV}\right]\\
&\leq& \epsilon_n. \label{PtoQ1}
\end{eqnarray}

Since by construction
$\mathbf{Q}^{(1)}_{X^nZ^nMM'\hat{M}'}=\mathbf{Q}^{(2)}_{X^nZ^nMM'\hat{M}'}$, 
\begin{eqnarray}
\Pbb_{\mathbf{Q}^{(1)}}[\hat{M}'\neq M']=\Pbb_{\mathbf{Q}^{(2)}}[\hat{M}'\neq M'].
\end{eqnarray}
Also, since $R'<I(V;Z)$, the codebook is randomly generated, and $M'$ is uniformly distributed under $Q$, it is well known that the maximum likelihood decoder $\Pbf_D$ (as well as a variety of other decoders) will drive the error probability to zero as $n$ goes to infinity. This can be seen from Fig. \ref{auxiliary}, by identifying, for fixed $M$, that $M'$ is the message to be transmitted over the memoryless channel $\Pbar_{Z|V}$. Therefore,
\begin{eqnarray}
\mathbb{E}_{\mathcal{C}^{(n)}}\left[\mathbb{P}_{\mathbf{Q}^{(1)}}[M'\neq \hat{M}']\right]\leq \delta_n\rightarrow_n 0.
\end{eqnarray}
Applying Lemma \ref{helper}, we obtain
\begin{eqnarray}
&&\Ebb_{\mathcal{C}^{(n)}} \left[\lVert \mathbf{Q}^{(1)}_{X^nZ^nM\hat{M}'}-\mathbf{Q}^{(2)}_{X^nZ^nMM'} \rVert_{TV}\right]\nonumber\\
&\leq&\mathbb{E}_{\mathcal{C}^{(n)}}\left[\Pbb_{\mathbf{Q}^{(1)}}[\hat{M}'\neq M']\right]
\leq\delta_n.
\end{eqnarray}
Thus by Property \ref{property-tv}$(\ref{d})$ and definitions $(\ref{Q1})$ and $(\ref{Q2})$,
\begin{eqnarray}
\Ebb_{\mathcal{C}^{(n)}} \left[\lVert \mathbf{Q}^{(1)}_{X^nZ^nM\hat{M}'Y^n}-\mathbf{Q}^{(2)}_{X^nZ^nMM'Y^n} \rVert_{TV}\right]\leq \delta_n. \label{Q1toQ2}
\end{eqnarray}
Combining $(\ref{PtoQ1})$ and $(\ref{Q1toQ2})$ and using Property \ref{property-tv}$(\ref{c})$ and \ref{property-tv}$(\ref{e})$, we have
\begin{eqnarray}
\Ebb_{\mathcal{C}^{(n)}}\left[ \lVert \mathbf{P}_{X^nY^n}-\mathbf{Q}^{(2)}_{X^nY^n}\rVert_{TV}\right]
\leq \epsilon_n+\delta_n\label{PtoQ2}
\end{eqnarray}
where $\epsilon_n$ and $\delta_n$ are the error terms introduced from the soft-covering lemma and channel coding, respectively.

Repeating the same steps as $(\ref{mm1})$ through $(\ref{n3})$ on $\Pbf$, $\mathbf{Q}^{(2)}$, and $\Pbar$, we obtain
\begin{eqnarray}
&&\Ebb_{\Ccal^{(n)}}\left[ \Ebb_{\mathbf{P}}[d(X^n,Y^n)]\right]\nonumber\\
&\leq&\Ebb_{\Pbar}\left[d(X,Y)\right]+d_{max}(\epsilon_n+\delta_n). \label{endp}
\end{eqnarray}
Taking the limit on both sides gives us
\begin{eqnarray}
\ls\Ebb_{\Ccal^{(n)}}\left[ \Ebb_{\mathbf{P}}[d(X^n,Y^n)]\right]\leq D.
\end{eqnarray}
Therefore, there exists a codebook satisfying the requirement. \QEDA

\section{The Berger-Tung Inner Bound} \label{bt}
The application of the likelihood encoder can go beyond single-user communications. In this section, we will demonstrate the use of the likelihood encoder via an alternative proof for achieving the Berger-Tung inner bound for the problem of multi-terminal source coding. Notice that no Markov lemma is needed in this proof. Similar to the single-user case, the key is to identify an auxiliary distribution that has nice properties and show that the system-induced distribution and the auxiliary distribution we choose are close in total variation. 

\subsection{Problem Setup and Result Review}
We now consider a pair of correlated sources $({X_1}^n,{X_2}^n)$, distributed i.i.d. according to $({{X_1}}_t,{X_2}_t)\sim P_{{X_1}{X_2}}$, independent encoders, and a joint decoder, satisfying the following constraints:
\begin{itemize}
\item Encoder 1 ${f_1}_n: {\mathcal{X}_1}^n \mapsto \mathcal{M}_1$ (possibly stochastic);
\item Encoder 2 ${f_2}_n: {\mathcal{X}_2}^n \mapsto \mathcal{M}_2$ (possibly stochastic);
\item Decoder $g_n: \mathcal{M}_1\times \mathcal{M}_2 \mapsto {\mathcal{Y}_1}^n\times {\mathcal{Y}_2}^n$ (possibly stochastic);
\item Compression rates: $R_1, R_2$, i.e. $|\mathcal{M}_1|=2^{nR_1}$, $|\mathcal{M}_2|=2^{nR_2}$.
\end{itemize}
The system performance is measured according to the time-averaged distortion (as defined in the notation section): 
\begin{itemize}
\item $d_1({X_1}^n, {Y_1}^n) =\frac1n\sum_{t=1}^n d_1({X_1}_t,{Y_1}_t),$
\item $d_2({X_2}^n, {Y_2}^n) =\frac1n\sum_{t=1}^n d_2({X_2}_t,{Y_2}_t),$

where $d_1(\cdot,\cdot)$ and $d_2(\cdot,\cdot)$ can be different distortion measures.
\end{itemize}
\begin{defn}
$(R_1,R_2)$ is achievable under distortion level $(D_1,D_2)$ if there exists a sequence of rate $(R_1,R_2)$ encoders and decoder $({f_1}_n, {f_2}_n, g_n)$ such that
\begin{eqnarray}
\ls\mathbb{E}[d_1({X_1}^n, {Y_1}^n)] \leq D_1,\\
\ls\mathbb{E}[d_2({X_2}^n, {Y_2}^n)] \leq D_2.
\end{eqnarray}
\end{defn}

The above mathematical formulation is illustrated in Fig. \ref{setup-bt}.

The achievable rate region is not yet known in general. But an inner bound, reproduced below, was given in \cite{tung} and \cite{berger1977} and is known as the Berger-Tung inner bound. The rates $(R_1,R_2)$ are achievable if
\begin{eqnarray}
R_1&>& I(X_1;U_1|U_2), \label{rate1}\\
R_2&>& I(X_2;U_2|U_1), \label{rate2}\\
R_1+R_2&>& I(X_1,X_2;U_1,U_2) \label{rate12}
\end{eqnarray}
for some $P_{U_1X_1X_2U_2}=P_{X_1X_2}P_{U_1|X_1}P_{U_2|X_2}$, and functions $\phi_k(\cdot,\cdot)$ such that $\mathbb{E}[d_k(X_k,Y_k)]\leq D_k$, where $Y_k\triangleq \phi_k(U_1,U_2), k=1,2$. \footnote{This region, after optimizing over auxiliary variables, is in general not convex, so it can be improved to the convex hull through time-sharing.}

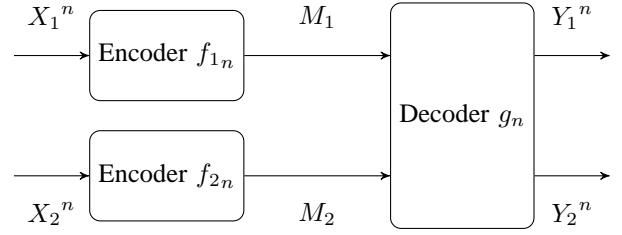
\begin{figure}
  \centering
\begin{tikzpicture}
[node distance=1cm,minimum height=12mm,minimum width=14mm,arw/.style={->,>=stealth'}]
  \node[coordinate] (source) {};
  \node[coordinate] (x1)[yshift=8mm] {};
  \node[coordinate] (x2)[yshift=-8mm] {};
  \node[rectangle,draw,rounded corners] (encoder1) [right =of x1] {Encoder ${f_1}_n$};
  \node[rectangle,draw,rounded corners] (encoder2) [right =of x2] {Encoder ${f_2}_n$};  
  
  \node[rectangle,draw,rounded corners, minimum height=30mm] (decoder) [right =5cm of source] {Decoder $g_n$};
  \node[coordinate](y1) [right =of decoder, yshift=8mm] {};
  \node[coordinate](y2) [right =of decoder, yshift=-8mm] {};
    
  \draw [arw] (x1) to node[above, yshift=-1mm]{${X_1}^n$} (x1 -| encoder1.west);
  \draw [arw] (x2) to node[below, yshift=1mm]{${X_2}^n$} (x2 -| encoder2.west);
  \draw [arw] (encoder1) to node[above, yshift=-1mm]{$M_1$} (encoder1 -| decoder.west);
  \draw [arw] (encoder2) to node[below, yshift=1mm]{$M_2$} (encoder2 -| decoder.west);
   
  \draw [arw] (decoder.east |- y1) to node[above,yshift=-1mm]{${Y_1}^n$} (y1);
  \draw [arw] (decoder.east |- y2) to node[below,yshift=1mm]{${Y_2}^n$} (y2);
\end{tikzpicture}
\caption{Berger-Tung problem setup}
\label{setup-bt}
\end{figure}

\subsection{Achievability Proof Using the Likelihood Encoder} \label{bt-proof}
We keep the same convention of using $P$ to denote the system-induced distribution and using $\Pbar$ for the distribution selected to optimize $(\ref{rate1})$-$(\ref{rate12})$ and any marginal or conditional distributions derived from it. Notice that $\Pbar_{X_1X_2}=P_{X_1X_2}$. We use $\Pbar_{{U_1}^n{X_1}^n{X_2}^n{U_2}^n}$ to denote the i.i.d. distribution, i.e.
\begin{eqnarray}
\Pbar_{{U_1}^n{X_1}^n{X_2}^n{U_2}^n}=\prod_{t=1}^n\Pbar_{U_1X_1X_2U_2}.
\end{eqnarray}

For simplicity, we focus on the corner points, $C_1\triangleq \left(I_{\Pbar}(X_1;U_1),I_{\Pbar}(X_2;U_2|U_1)\right)$ and $C_2\triangleq \left(I_{\Pbar}(X_1;U_1|U_2),I_{\Pbar}(X_2;U_2)\right)$, of the region given in $(\ref{rate1})$ through $(\ref{rate12})$ and use convexity to claim the complete region. Below we demonstrate how to achieve $C_1$. The point $C_2$ follows by symmetry.

\subsubsection{High-level outline}
Fix a $\Pbar_{U_1U_2|X_1X_2}=\Pbar_{U_1|X_1}\Pbar_{U_2|X_2}$ and functions ${\phi_k}(\cdot,\cdot)$ such that $Y_k={\phi_k}(U_1,U_2)$ and $\Ebb_{\Pbar}\left[d_k(X_k,Y_k)\right]<D_k$.  Note that $U_1\inout X_1\inout X_2\inout U_2$ forms a Markov chain under $\Pbar$. We must show that any rate pair $(R_1,R_2)$ satisfying $R_1>I_{\Pbar}(X_1;U_1)$ and $R_2>I_{\Pbar}(X_2;U_2|U_1)$ is achievable.

As expected, the decoder will use a lossy representation of one source as side information to assist reconstruction of the other source. We can choose an $R_2'<I_{\Pbar}(U_1;U_2)$ such that $R_2+R_2'>I_{\Pbar}(X_2;U_2)$. Here $R_2'$ corresponds to the rate of a virtual message $M_2'$ which is produced by Encoder 2 but not physically transmitted to the receiver. This will play the role of the index of the codeword in the bin in a traditional covering and random-binning proof. 

First we use the likelihood encoder derived from $\Pbar_{X_1U_1}$ and a random codebook $\{{u_1}^n(m_1)\}$ generated according to $\Pbar_{U_1}$ for Encoder 1. Then we use the likelihood encoder derived from $\Pbar_{X_2U_2}$ and another random codebook $\{{u_2}^n(m_2,m_2')\}$ generated according to $\Pbar_{U_2}$ for Encoder 2. The decoder uses the transmitted message $M_1$ to decode ${U_1}^n$, as in the point-to-point case,  and uses the transmitted message $M_2$ along with the decoded ${U_1}^n$ to decode $M_2'$ as $\hat{M}_2'$, as in the Wyner-Ziv case, and reproduces $u_2^n(M_2,\hat{M}_2')$. Finally, the decoder outputs the reconstructions ${Y_k}^n$ according to the symbol-by-symbol functions $\phi_k(\cdot,\cdot)$ of ${U_1}^n$ and ${U_2}^n$.

The distribution induced by the sources, the encoders and decoder is
\begin{eqnarray}
P_{{X_1}^n{X_2}^n{U_1}^nM_1M_2M_2'\hat{M}_2'{Y_1}^n{Y_2}^n}
=P_{{X_1}^n{X_2}^n}\Pbf_1\Pbf_2 \label{pjoint}
\end{eqnarray}
where
\begin{eqnarray}
&&\Pbf_1(m_1,{u_1}^n|{x_1}^n)\nonumber\\
&\triangleq& \Pbf_{M_1|{X_1}^n}(m_1|{x_1}^n)\Pbf_{{U_1}^n|M_1}({u_1}^n|m_1)\\
&\triangleq& \Pbf_{LE1}(m_1|{x_1}^n) \Pbf_{D1}({u_1}^n|m_1)
\end{eqnarray}
and
\begin{eqnarray}
&&\Pbf_2(m_2,m_2',\hat{m}_2',{y_1}^n,{y_2}^n|{x_2}^n,{u_1}^n)\nonumber\\
&\triangleq& \Pbf_{M_2M_2'|{X_2}^n}(m_2,m_2'|{x_2}^n)\Pbf_{\hat{M}_2'|M_2{U_1}^n}(\hat{m}_2'|m_2,{u_1}^n)\nonumber\\
&&\prod_{k=1,2}P_{{Y_k}^n|{U_1}^nM_2\hat{M}_2'}({y_k}^n|{u_1}^n,m_2,\hat{m}_2')\\
&\triangleq& \Pbf_{LE2}(m_2,m_2'|{x_2}^n) \Pbf_{D2}(\hat{m}_2'|m_2,{u_1}^n)\nonumber\\
&&\prod_{k=1,2} \Pbf_{\Phi, k}({y_k}^n|{u_1}^n,m_2,\hat{m}_2'), \label{p2}
\end{eqnarray}
where $\Pbf_{LE1}$ and $\Pbf_{LE2}$ are the likelihood encoders; $\Pbf_{D1}$ is the first part of the decoder that does a codeword lookup on $\Ccal_1^{(n)}$; $\Pbf_{D2}$ is the second part of the decoder that decodes $m_2'$ as $\hat{m}_2'$; and $\Pbf_{\Phi,k}({y_k}^n|{u_1}^n,m_2,\hat{m}_2')$ is the third part of the decoder that reconstructs the source sequences.

The analysis mimics the point-to-point analysis (Section \ref{p2p}) for $\Pbf_1$ and the Wyner-Ziv analysis (Section \ref{sec-wz}) for $\Pbf_2$.

\subsubsection{Proof}
We now restate the behavior of the encoders and decoder -- components of the induced distribution stated in $(\ref{pjoint})$-$(\ref{p2})$. These are derived from the distribution $\Pbar_{U_1X_1X_2U_2}$ and $\phi_1(\cdot, \cdot)$ and $\phi_2(\cdot, \cdot)$ stated in the outline.

{\textbf{Codebook generation}}: We independently generate $2^{nR_1}$ sequences in ${\mathcal{U}_1}^n$ according to $\prod_{t=1}^n\Pbar_{U_1}({u_1}_t)$ and index them by $m_1\in\{1, \ldots ,2^{nR_1}\}$, and independently generate $2^{n(R_2+R_2')}$ sequences in ${\mathcal{U}_2}^n$ according to $\prod_{t=1}^n\Pbar_{U_2}({u_2}_t)$ and index them by $(m_2,m_2')\in\{1, \ldots ,2^{nR_2}\}\times\{1, \ldots, 2^{nR_2'}\}$. We use $\Ccal_1^{(n)}$ and $\Ccal_2^{(n)}$ to denote the two random codebooks, respectively. 

{\textbf{Encoders}}: The first encoder $\Pbf_{LE1}(m_1|{x_1}^n)$ is the likelihood encoder according to $\Pbar_{{X_1}^n{U_1}^n}$ and $\Ccal_1^{(n)}$. The second encoder $\Pbf_{LE2}(m_2,m_2'|{x_2}^n)$ is the likelihood encoder according to $\Pbar_{{X_2}^n{U_2}^n}$ and $\Ccal_2^{(n)}$. The first encoder sends $M_1$ and the second encoder sends $M_2$.

{\textbf{Decoder}}:  First, let $\Pbf_{D1}({u_1}^n|m_1)$ be a $\Ccal_1^{(n)}$ codeword lookup decoder. Then, let $\Pbf_{D2}(\hat{m}_2'|m_2,{u_1}^n)$ be a good channel decoder with respect to the sub-codebook $\Ccal_2^{(n)}(m_2)=\{{u_2}^n(m_2,a)\}_a$ and the memoryless channel $\Pbar_{U_1|U_2}$. Last, define ${\phi_k}^n({u_1}^n,{u_2}^n)$ as the concatenation $\{{\phi_k}({u_1}_t,{u_2}_t)\}_{t=1}^n$ and set the decoders $\Pbf_{\Phi,k}$ to be the deterministic functions
\begin{eqnarray}
&&\Pbf_{\Phi,k}({y_k}^n|{u_1}^n,m_2,\hat{m}_2')\nonumber\\
&\triangleq&\mathbbm{1}\{{y_k}^n={{\phi_k}^n({u_1}^n,{U_2}^n(m_2,\hat{m}_2'))}\}.
\end{eqnarray}

{\textbf{Analysis}}: We need the following distributions: the induced distribution $\Pbf$ and auxiliary distributions $\Qbf_1$ and $\Qbf_1^*$. Encoder 1 makes $\Pbf$ and $\Qbf_1$ close in total variation. Distribution $\Qbf_1^*$ (random only with respect to the second codebook $\Ccal_2^{(n)}$) is the expectation of $\Qbf_1$ over the random codebook $\Ccal_1^{(n)}$. This is really the key step in the proof. By considering the expectation of the distribution with respect to $\Ccal_1^{(n)}$, we effectively remove Encoder 1 from the problem and turn the message from Encoder 1 into memoryless side information at the decoder. Hence, the two distortions (averaged over $\Ccal_1^{(n)}$) under $\Pbf$ are roughly the same as the distortions under $\Qbf_1^*$, which is a much simpler distribution. 
We then recognize $\Qbf_1^*$ as precisely $\Pbf$ in $(\ref{jointPP2})$ from the Wyner-Ziv proof of the previous section, with a source pair $(X_1,X_2)$, a pair of reconstructions $(Y_1,Y_2)$ and $U_1$ as the side information.

a) The auxiliary distribution $\Qbf_1$ takes the following form:
\begin{eqnarray}
{\mathbf{Q}_1}_{{X_1}^n{X_2}^n{U_1}^nM_1M_2M_2'\hat{M}_2'{Y_1}^n{Y_2}^n}={\mathbf{Q}_1}_{M_1{U_1}^n{X_1}^n {X_2}^n}\Pbf_2
\end{eqnarray}
where
\begin{eqnarray}
&&{\mathbf{Q}_1}_{M_1{U_1}^n{X_1}^n {X_2}^n}(m_1,{u_1}^n,{x_1}^n,{x_2}^n)\nonumber\\
&=&\frac{1}{2^{nR_1}}\mathbbm{1}\{{u_1}^n={U_1}^n(m_1)\}\Pbar_{{X_1}^n|{U_1}^n}({x_1}^n|{u_1}^n)\nonumber\\
&&\Pbar_{{X_2}^n|{X_1}^n}({x_2}^n|{x_1}^n).\label{qqq}
\end{eqnarray}
Note that $\Qbf_1$ is the idealized distribution with respect to the first message, as introduced in the point-to-point case. Hence, by the same arguments, since $R_1>I_{\Pbar}(X_1;U_1)$, using the soft-covering lemma, 
\begin{eqnarray}
\mathbb{E}_{\Ccal_1^{(n)}}\left[\lVert {\mathbf{Q}_1}-{\mathbf{P}} \rVert_{TV}\right]\leq {\epsilon_1}_n, \label{tv-layer1}
\end{eqnarray}
where $\Qbf_1$ and $\Pbf$ are distributions over random variables ${X_1}^n,{X_2}^n,{U_1}^n,M_1,M_2,M_2',\hat{M}_2',{Y_1}^n,{Y_2}^n$ and ${\epsilon_1}_n$ is the error term introduced from soft-covering lemma.

b) Taking the expectation over codebook $\Ccal_1^{(n)}$, we define
\begin{eqnarray}
&&{\Qbf^*_1}_{{X_1}^n{X_2}^n{U_1}^nM_2M_2'\hat{M}_2'{Y_1}^n{Y_2}^n}\nonumber\\
&\triangleq& \mathbb{E}_{\Ccal_1^{(n)}}\left[{\mathbf{Q}_1}_{{X_1}^n{X_2}^n{U_1}^nM_2M_2'\hat{M}_2'{Y_1}^n{Y_2}^n}\right].
\end{eqnarray}

Note that under this definition of $\Qbf^*_1$, we have
\begin{eqnarray}
&&{\Qbf^*_1}_{{X_1}^n{X_2}^n{U_1}^nM_2M_2'\hat{M}_2'{Y_1}^n{Y_2}^n}\nonumber\\
&&({x_1}^n,{x_2}^n,{u_1}^n,m_2,m_2',\hat{m}_2',{y_1}^n,{y_2}^n)\nonumber\\
&=&\mathbb{E}_{\Ccal_1^{(n)}}\left[{\Qbf^*_1}_{{X_1}^n{X_2}^n{U_1}^n}({x_1}^n,{x_2}^n,{u_1}^n)\right]\nonumber\\
&&\Pbf_2(m_2,m_2',\hat{m}_2',{y_1}^n,{y_2}^n|{x_2}^n,{u_1}^n) \\
&=&\Pbar_{{X_1}^n{X_2}^n{U_1}^n}({x_1}^n,{x_2}^n,{u_1}^n)\nonumber\\
&&\Pbf_2(m_2,m_2',\hat{m}_2',{y_1}^n,{y_2}^n|{x_2}^n,{u_1}^n),
\end{eqnarray}
where the last step can be verified using the same technique as $(\ref{expectation})$ given in Section \ref{p2p}.

By Property \ref{property-tv}$(\ref{b})$,
\begin{eqnarray}
&&\mathbb{E}_{\Ccal_1^{(n)}}\left[ \Pbb_{\Pbf}\left[\left\{d_k({X_k}^n,{Y_k}^n)>D_k\right\}\right]\right]\nonumber\\
&=&\mathbb{E}_{\Ccal_1^{(n)}}\left[ \Ebb_{\mathbf{P}}\left[\mathbbm{1}\left\{d_k({X_k}^n,{Y_k}^n)>D_k\right\}\right]\right]\\
&\leq&\mathbb{E}_{\Ccal_1^{(n)}} \left[\Ebb_{\mathbf{Q}_1}[\mathbbm{1}\left\{d_k({X_k}^n,{Y_k}^n)>D_k\right\}]\right]+{\epsilon_1}_n\\
&=&\Ebb_{\Ccal_1^{(n)}}\bigg[\sum_{{x_k}^n,{y_k}^n}\Qbf_1({x_k}^n,{y_k}^n)\nonumber\\
&&\mathbbm{1}\left\{d_k({X_k}^n,{Y_k}^n)>D_k\bigg\}\right]+{\epsilon_1}_n\\
&=&\sum_{{x_k}^n,{y_k}^n}\Ebb_{\Ccal_1^{(n)}}[\Qbf_1({x_k}^n,{y_k}^n)]\nonumber\\
&&\mathbbm{1}\left\{d_k({X_k}^n,{Y_k}^n)>D_k\right\}+{\epsilon_1}_n\\
&=&\sum_{{x_k}^n,{y_k}^n}\Qbf_1^*({x_k}^n,{y_k}^n)\nonumber\\
&&\mathbbm{1}\left\{d_k({X_k}^n,{Y_k}^n)>D_k\right\}+{\epsilon_1}_n\\
&=&\Pbb_{\Qbf_1^*}\left[d_k({X_k}^n,{Y_k}^n)>D_k\right]+{\epsilon_1}_n. \label{rr1}
\end{eqnarray}

Note that $\Qbf^*_1$ is exactly of the form of the induced distribution $\Pbf$ in the Wyner-Ziv proof of the previous section, with the inconsequential modification that there are two reconstructions and two distortion functions, and working with the indicator functions instead of the distortion functions as in Section \ref{sec-excess}. Thus, by $(\ref{Q1})$ through $(\ref{endp})$, we obtain 
\begin{eqnarray}
&&\Ebb_{\Ccal_2^{(n)}} \left[\Pbb_{\mathbf{Q}^*_1}\left[d_k({X_k}^n,{Y_k}^n)>D_k\right]\right] \nonumber\\
&=&\Ebb_{\Ccal_2^{(n)}} \left[\Ebb_{\mathbf{Q}^*_1}\left[\mathbbm{1}\left\{d_k({X_k}^n,{Y_k}^n)>D_k\right\}\right]\right]\\
&\leq& \Ebb_{\Pbar}\left[\mathbbm{1}\left\{d_k({X_k}^n,{Y_k}^n)>D_k\right\}\right]+({\epsilon_2}_n+\delta_n)\\
&=& \Pbb_{\Pbar}[d_k({X_k}^n,{Y_k}^n)>D_k]+({\epsilon_2}_n+\delta_n),\label{DQstar}
\end{eqnarray}
where ${\epsilon_2}_n$ and $\delta_n$ are error terms introduced from the soft-covering lemma and channel decoding, respectively.

Combining $(\ref{rr1})$ and $(\ref{DQstar})$,
\begin{eqnarray}
&&\mathbb{E}_{\Ccal_2^{(n)}}\left[\mathbb{E}_{\Ccal_1^{(n)}}\left[\Pbb_{\Pbf}\left[\left\{d_k({X_k}^n,{Y_k}^n)>D_k\right\}\right]\right]\right]\nonumber\\
&\leq&\Ebb_{\Ccal_2^{(n)}} \left[\Pbb_{\mathbf{Q}^*_1}\left[d_k({X_k}^n,{Y_k}^n)>D_k\right]\right] +{\epsilon_1}_n\label{br1}\\
&\leq& \Pbb_{\Pbar}[d_k({X_k}^n,{Y_k}^n)>D_k]+({\epsilon_1}_n+{\epsilon_2}_n+\delta_n)\label{qqqq}
\end{eqnarray}
where $(\ref{br1})$ follows from $(\ref{rr1})$; $(\ref{qqqq})$ follows from $(\ref{rr1})$ and $(\ref{DQstar})$. 

Consequently,
\begin{eqnarray}
&&\Ebb_{\Ccal^{(n)}}[\Pbb_{\Pbf}[d_1({X_1}^n,{Y_1}^n)>D_1 \text{ or }\nonumber\\
&&\ \ \ \ \ \ \ \ \ \ \ \ d_2({X_2}^n,{Y_2}^n)>D_2]]\nonumber\\
&\leq&\Ebb_{\Ccal^{(n)}}\left[\sum_{k=1,2}\Pbb_{\Pbf}\left[d_k({X_k}^n,{Y_k}^n)>D_k\right]\right]\label{ub}\\
&=&\sum_{k=1,2}\Ebb_{\Ccal^{(n)}}\left[\Pbb_{\Pbf}\left[d_k({X_k}^n,{Y_k}^n)>D_k\right]\right]\\
&\leq&\sum_{k=1,2}\Pbb_{\Pbar}[d_k({X_k}^n,{Y_k}^n)>D_k]\nonumber\\
&&+2({\epsilon_1}_n+{\epsilon_2}_n+\delta_n)\\
&=&\sum_{k=1,2}\Pbb_{\Pbar}\left[\sum_{t=1}^nd_k({X_k}_t,{Y_k}_t)>D_k\right]\nonumber\\
&&+2({\epsilon_1}_n+{\epsilon_2}_n+\delta_n)\\
&\triangleq&\epsilon_n \rightarrow_n 0 \label{lln}
\end{eqnarray}
where $(\ref{ub})$ follows from the union bound and $(\ref{lln})$ follows from the law of large numbers. 

Therefore, there exists a codebook under which
\begin{eqnarray}
\Pbb_P\left[d_1({X_1}^n,{Y_1}^n)>D_1 \text{ or }d_2({X_2}^n,{Y_2}^n)>D_2\right]\leq \epsilon_n
\end{eqnarray}
which completes the proof under excess distortion. To get the bounds under average distortion, note that
\begin{eqnarray}
&&\Ebb_P\left[d_k({X_k}^n,{Y_k}^n)\right]\nonumber\\
&\leq& D_k \Pbb_P[d_1({X_1}^n,{Y_1}^n)\leq D_k]\nonumber\\
&&+ {d_k}_{max}\Pbb_P[d_1({X_1}^n,{Y_1}^n)>D_k]\\
&\leq&D_k+{d_k}_{max}\epsilon_n.
\end{eqnarray}
\QEDA

\begin{rem}
Note that the proof above uses the proof of Wyner-Ziv achievability from the previous section. To do the entire proof step by step, we would define a total of three auxiliary distributions, which would be the $\Qbf_1$ used in the proof, as well as  $\mathbf{Q}_2^{(1)}$ and $\mathbf{Q}_2^{(2)}$ defined below for completeness. The steps outlined above show how to relate the induced distribution $\Pbf$ to $\Qbf_1$ and its expectation $\Qbf_1^*$. This effectively converts the message from Encoder 1 into memoryless side information at the decoder. The omitted steps, as seen in the previous section, relate $\Qbf_1^*$ to $\mathbf{Q}_2^{(1)}$ through the soft-covering lemma and $\mathbf{Q}_2^{(1)}$ to $\mathbf{Q}_2^{(2)}$ through reliable channel decoding. The expected value of $\mathbf{Q}_2^{(2)}$ over codebooks is the desired distribution $\Pbar$. For reference, the omitted auxiliary distributions are
\begin{eqnarray}
&&{\mathbf{Q}_2}_{M_2M_2'{U_2}^n{X_2}^n{X_1}^n{U_1}^n}\nonumber\\
&=&\frac{1}{2^{n(R_2+R_2')}}\mathbbm{1}\{{u_2}^n={U_2}^n(m_2,m_2')\}\Pbar_{{X_2}^n|{U_2}^n}({x_2}^n|{u_2}^n)\nonumber\\
&&\Pbar_{{X_1}^n{U_1}^n|{X_2}^n}({x_1}^n,{u_1}^n|{x_2}^n),
\end{eqnarray}
which is of the same structure as the idealized distribution described in Fig. \ref{auxiliary}, and
\begin{eqnarray}
{\mathbf{Q}_2^{(1)}}_{{X_1}^n{X_2}^n{U_1}^nM_2M_2'\hat{M}_2'{Y_1}^n{Y_2}^n}\triangleq{\mathbf{Q}_2}_{{X_1}^n{X_2}^n{U_1}^nM_2M_2'}\nonumber\\
\Pbf_D(\hat{m}_2'|m_2,{u_1}^n)\prod_{k=1,2}\Pbf_{\Phi,k} ({y_k}^n|{u_1}^n,m_2,\hat{m}_2')\\
{\mathbf{Q}_2^{(2)}}_{{X_1}^n{X_2}^n{U_1}^nM_2M_2'\hat{M}_2'{Y_1}^n{Y_2}^n}\triangleq{\mathbf{Q}_2}_{{X_1}^n{X_2}^n{U_1}^nM_2M_2'}\nonumber\\
\Pbf_D(\hat{m}_2'|m_2,{u_1}^n)\prod_{k=1,2}\Pbf_{\Phi,k} ({y_k}^n|{u_1}^n,m_2,m_2').
\end{eqnarray}
\end{rem}

\begin{rem}
For comparison with the traditional joint typicality encoder proof, recall from \cite{network-it} that to bound the different error events, we would need the regular covering lemma, the conditional typicality lemma, the Markov lemma, and the mutual packing lemma, some of which are quite involved to verify. With the likelihood encoder, all we need is the soft-covering lemma and Lemma \ref{helper}. 
\end{rem}

\section{Non-asymptotic Analysis} \label{non-asymptotic}
In this section, we analyze the non-asymptotic performance of the likelihood encoder by evaluating how fast the excess distortion approaches zero. For brevity, we demonstrate the analysis only for the point-to-point case. 

Let the achievable rate-distortion region $\mathcal{R}$ be
\begin{eqnarray}
\mathcal{R}\triangleq \{(R,D): R>R(D)\}.
\end{eqnarray}

For a fixed $(R,D)\in\mathcal{R}$, we aim to minimize the probability of excess distortion (from Section \ref{sec-excess}), using a random codebook and the likelihood encoder, over valid choices of $\Pbar_{Y|X}$, and evaluate how fast the excess distortion decays with blocklength $n$ under the optimal $\Pbar_{Y|X}$. Mathematically, we want to obtain
\begin{eqnarray}
\inf_{\Pbar_{Y|X}}\Ebb_{\mathcal{C}^n}\left[\Pbb_{\Pbf}\left[d(X^n,Y^n)>D\right]\right],\label{opt}
\end{eqnarray}
where the subscript $\Pbf$ indicates probability taken with respect to the induced distribution.

To evaluate how fast the probability of excess distortion approaches zero, note in $(\ref{ww1})$ that the first term is governed (approximately) by the gap $D-\Ebb_{\Pbar}[d(X,Y)]$ and the second term is governed (approximately) by the the gap $R-I_{\Pbar}(X;Y)$. To see this, observe that for any $\beta>0$, 
\begin{eqnarray}
\epsilon_n'&\triangleq&\Pbb_{\Pbar}[d(X^n,Y^n)>D]\nonumber\\
&=&\Pbb_{\Pbar}\left[\frac1n\sum_{t=1}^n d(X_t,Y_t)>D\right]\label{cher1}\\
&\leq&\inf_{\beta>0}\left[\frac{\Ebb_{\Pbar}[2^{\beta d(X,Y)}]}{2^{\beta D}}\right]^n\label{cher}\\
&=&\exp\left(-n\log\left(\inf_{\beta>0}\Ebb_{\Pbar}\left[2^{\beta (d(X,Y)-D)}\right]\right)^{-1}\right)\label{jen}\\
&=& \exp\left(-n\eta(\Pbar_{Y|X})\right)\label{cher2}
\end{eqnarray}
where $(\ref{cher})$ follows from the Chernoff bound and we have implicitly defined 
\begin{eqnarray}
\eta(\Pbar_{Y|X})\triangleq \log \left(\inf_{\beta>0}\Ebb_{\Pbar}\left[2^{\beta (d(X,Y)-D)}\right]\right)^{-1}. \label{eta}
\end{eqnarray}
An upper bound on the second term in $(\ref{ww1})$ is given in \cite{cuff2012distributed}, restated below:
\begin{eqnarray}
\epsilon_n\leq \frac32 \exp\left(-n\gamma(\Pbar_{Y|X}) \right),
\end{eqnarray}
where 
\begin{eqnarray}
\gamma(\Pbar_{Y|X})
\triangleq \max_{\alpha\geq1, \alpha'\leq2}\frac{\alpha-1}{2\alpha-\alpha'}\bigg(R-\check{I}_{\Pbar,\alpha}(X;Y)\nonumber\\
+(\alpha'-1)(\check{I}_{\Pbar,\alpha}(X;Y)-\bar{I}_{\Pbar,\alpha'}(X;Y))\bigg)\label{gamma}
\end{eqnarray}
\small
\begin{eqnarray}
\check{I}_{\Pbar,\alpha}(X;Y)
\triangleq \frac{1}{\alpha-1}\log\left(\Ebb_{\Pbar}\left[\left(\frac{\Pbar_{X,Y}(X,Y)}{\Pbar_{X}(X)\Pbar_{Y}(Y)}\right)^{\alpha-1}\right]\right)
\end{eqnarray}
\normalsize
\begin{eqnarray}
\bar{I}_{\Pbar,\alpha'}(X,Y)
&\triangleq& \frac{1}{\alpha'-1}\log\left(\left(\Ebb_{\Pbar_X}\left[\Gamma\right]\right)^2\right)
\end{eqnarray}
\begin{eqnarray}
\Gamma&\triangleq&\sqrt{\Ebb_{\Pbar_{Y|X}}\left[\left(\frac{\Pbar_{XY}(X,Y)}{\Pbar_{X}(X)\Pbar_Y(Y)}\right)^{\alpha'-1}\right]}.
\end{eqnarray}

Both $\epsilon_n'$ and $\epsilon_n$ decay exponentially with $n$. To obtain an upper bound on the excess distortion given in $(\ref{opt})$, we now have a new optimization problem in the following form:
\begin{eqnarray}
\inf_{{\Pbar}_{Y|X}}\left[ \exp\left(-n\eta(\Pbar_{Y|X})\right)+\frac3 2 \exp\left(-n\gamma(\Pbar_{Y|X})\right)\right], \label{opt-sum}
\end{eqnarray}
where $\eta(\Pbar_{Y|X})$ and $\gamma(\Pbar_{Y|X})$ are defined in $(\ref{eta})$ and $(\ref{gamma})$. Note that only choices of ${\Pbar}_{Y|X}$ such that $\Ebb_{\Pbar}[d(X,Y)]<D$ and $I_{\Pbar}(X;Y)<R$ should be considered for the optimization, as other choices render the bound degenerate. 

We can relax $(\ref{opt-sum})$ to obtain a simple upper bound on the excess distortion $\Pbb_{P}[d(X^n,Y^n)>D]$ given in the following theorem.
\begin{thm} \label{thm-sc}
The excess distortion $\Pbb_{P}[d(X^n,Y^n)>D]$ using the likelihood encoder is upper bounded by
\begin{eqnarray}
\inf_{\Pbar_{Y|X}} \frac5 2 \exp\left(-n\min\left\{\eta\left(\Pbar_{Y|X}\right),\gamma \left(\Pbar_{Y|X}\right)\right\}\right)
\end{eqnarray}
where $\eta(\Pbar_{Y|X})$ and $\gamma(\Pbar_{Y|X})$ are given in $(\ref{eta})$ and $(\ref{gamma})$, respectively.
\end{thm}

\begin{rem}
Note that this bound does not achieve Marton's source coding exponent that we know to be optimal \cite{marton},  \cite[Theorem 9.5]{ck2001} for rate-distortion theory. It may very well be that the likelihood encoder does not achieve the optimal exponent, though it may also be an artifact of our proof or the bound for the soft-covering lemma.
\end{rem}

\section{Connection with Random Binning Based Proof} \label{binning}
The likelihood encoder proof technique is similar in many ways to the random binning based analysis approach presented in \cite{rb-yassaee}. In this section, we present the random binning based analysis for point-to-point lossy compression in a format that resembles the likelihood encoder based proof. Our presentation is different from the way the authors presented the scheme in \cite{rb-yassaee}, stating explicitly the behavior of the encoder, for easy comparison with the likelihood encoder and Section \ref{p2p}.

\subsection{The Proportional-Probability Encoder}
We start by defining a source encoder that looks very similar in form to a likelihood encoder defined in Section \ref{sub-le}. Like any other source encoder, a \textit{proportional-probability encoder} receives a sequence $x_1,...,x_n$ and produces an index $m\in \{1, \ldots, 2^{nR}\}$.  

A codebook is specified by a non-empty collection $\Ccal$ of sequences $y^n\in \Ycal^n$ and indices $m(y^n)$ assigned to each $y^n \in \Ycal^n$.  The codebook and a joint distribution $P_{XY}$ specify the proportional-probability encoder.

Let $\Gcal(m|x^n)$ be the probability, as a result of passing $x^n$ through a memoryless channel given by $P_{Y|X}$, of finding $Y^n$ in the collection $\Ccal$ and retrieving the index $m$ from the codebook:
\begin{eqnarray}
&&\Gcal(m|x^n) \nonumber\\
& \triangleq &
\Pbb_{\prod P_{Y|X}} \left[ Y^n \in \Ccal, m(Y^n) = m \mid X^n=x^n \right] \\
& = & \sum_{y^n\in\Ccal}{P_{Y^n|X^n}(y^n|x^n)\mathbbm{1}\{m(y^n)=m\}}.
\end{eqnarray}

A proportional-probability encoder is a stochastic encoder that determines the message index with probability proportional to $\Gcal(m|x^n)$, i.e. 
\begin{equation}
P_{M|X^n}(m|x^n)=\frac{\Gcal(m|x^n)}{\sum_{m'\in\{1, \ldots, 2^{nR}\}}\Gcal(m'|x^n)}\propto \Gcal(m|x^n).
\label{def:ppe}
\end{equation}

Notice that the proportional-probability encoder and the likelihood encoder both behave stochastically with probability proportional to that of a memoryless channel. However, the channels are the reverse direction from each other. We will see that the codebook construction also differs slightly between the two proof techniques.

\subsection{Scheme Using the Proportional-Probability Encoder} \label{ppe}
Before going into the achievability scheme, we first state a lemma that will be used in the analysis.
\begin{lem} [Independence of random binning - \cite{rb-yassaee}, Theorem~1]  \label{conditional-rate}
Given a probability mass function $P_{XY}$, and each $y^n\in \Ycal^n$ is independently assigned to a bin index $b\in \{1, \ldots ,2^{nR_b}\}$ uniformly at random, where $B(y^n)$ denotes this random assignment. Define the joint distribution 
\begin{eqnarray}
\Pbf_{X^nY^nB}(x^n,y^n,b)\triangleq \prod_{i=1}^n P_{XY}(x_i,y_i)\mathbbm{1}\{B(y^n)=b\}.
\end{eqnarray}
If $R_b<H(Y|X)$, then we have
\begin{eqnarray}
\Ebb_{\Bcal}\left[\left\Vert\Pbf_{X^nB}-P_{X^n}P^U_{B}\right\Vert_{TV}\right] \rightarrow_n 0, 
\end{eqnarray}
where $P^U_{B}$ is a uniform distribution on $\{1, \ldots, 2^{nR_b}\}$ and $\Ebb_{\Bcal}$ denotes expectation taken over the random binning.
\end{lem}

We now outline the encoding-decoding scheme based on the proportional-probability encoder. 

Fix a $\Pbar_{Y|X}$ that satisfies $\Ebb_{\Pbar}[d(X,Y)]<D$ and choose the rates $R$ and $R'$ to satisfy $R'<H_{\Pbar}(Y|X)$ and $R+R'>H_{\Pbar}(Y)$.

{\textbf{Codebook generation}}: Each $y^n\in \Ycal^n$ is randomly and independently assigned to the codebook $\Ccal$ with probability $2^{-nR'}$. Then, independent of the construction of $\Ccal$, each $y^n\in\Ycal^n$ is independently assigned uniformly at random to one of $2^{nR}$ bins indexed by $M$. 

{\textbf{Encoder}}: The encoder $\Pbf_{PPE}(m|x^n)$ is the proportional-probability encoder with respect to $\Pbar$.  Specifically, the encoder chooses $M$ stochastically according to \eqref{def:ppe}, with $\Gcal$ based on $\Pbar$ as follows:
\begin{eqnarray}
\Gcal(m|x^n)=\sum_{y^n\in\Ccal}\Pbar_{Y^n|X^n}(y^n|x^n)\mathbbm{1}\{m(y^n)=m\},
\end{eqnarray}
where $\Pbar_{Y^n|X^n}(y^n|x^n)=\prod_{t=1}^n\Pbar_{Y|X}(y_t|x_t).$

{\textbf{Decoder}}: The decoder $\Pbf_D(y^n|m)$ selects a $y^n$ reconstruction that is in $\Ccal$ and has index $m=M$.  There will usually be more than one such $y^n$ sequence, but rarely will there be more than one ``good'' choice, due to the rates used.  The decoder can choose the most probable $y^n$ sequence or the unique typical sequence, etc.  The proof in \cite{rb-yassaee} uses a ``mismatch stochastic likelihood coder'' (MSLC) \cite{yassaee2013non} \cite{slc}, which stochastically decodes $y^n$, but many decoders will achieve the desired result.

\begin{rem}
\label{rem-dec}
Intuitively, a decoder can successfully decode the sequence intended by the encoder since there are roughly $2^{nH_{\Pbar}(Y)}$ typical $y^n$ sequences, and the collection $\Ccal$ together with the binning index $M$ provides high enough rate $R'+R>H_{\Pbar}(Y)$ to uniquely identify the sequence.
\end{rem}

{\textbf{Analysis}}:
The above scheme specifies a system-induced distribution of the form
\begin{eqnarray}
\Pbf_{X^nMY^n}(x^n,m,y^n)=\Pbar_{X^n}\Pbf_{PPE}(m|x^n)\Pbf_D(y^n|m).
\end{eqnarray}

To analyze the above scheme, we start by replacing the codebook used for encoding and decoding with a set of codebooks.  Recall that the codebook consists of a collection $\Ccal$ and index assignments $m(y^n)$ that are both randomly constructed.  Now consider a set of $2^{nR'}$ collections $\{\Ccal_f\}_{f\in\{1, \ldots, 2^{nR'}\}}$, indexed by $f$, created by assigning each $y^n$ sequence in $\Ycal^n$ randomly to exactly one collection equiprobably.  From this we define a set of $2^{nR'}$ codebooks, one for each $f$, each one consisting of the collection $\Ccal_f$ and the common message index function $m(y^n)$.  We use $\Kcal$ to denote this set of random codebooks.

By this construction, the original random collection $\Ccal$ in the codebook used by the encoder and decoder is equivalent in probability to using the first codebook associated with $\Ccal_1$.  It is also equivalent to using a random codebook in the set, which is a point we will utilize shortly.  The purpose of defining multiple codebooks is to facilitate general proof tools associated with uniform random binning.

Here we summarize the proof given in \cite{rb-yassaee}. In addition to the system-induced random variables, we introduce a random variable $F$ which is uniformly distributed on the set $\{1, \ldots, 2^{nR'}\}$ and independent of $X^n$.  The variable $F$ selects the codebook to be used---everything else about the encoding and decoding remains the same.  We have noted that the behavior and performance of this system with multiple codebooks is equivalent to that of the actual encoding and decoding.  Nevertheless, we will formalize this connection in \eqref{common}.  For now, we refer to this new distribution that includes many codebooks as the pseudo induced distribution $\tilde{\Pbf}$.  According to $\tilde{\Pbf}$, there is a set of randomly generated codebooks, and the one for use is selected by $F$.

The pseudo induced distribution can be expressed in the following form:
\begin{eqnarray}
&&\tilde{\Pbf}_{FX^nMY^n}(f,x^n,m,y^n)\nonumber\\
&=&P_F(f)\Pbar_{X^n}(x^n)\Pbf_{PPE}(m|x^n,f)\Pbf_D(y^n|m,f).
\end{eqnarray}
We reiterate that
\begin{eqnarray}
\Pbf_{X^nMY^n} \, {\buildrel D \over =} \, \tilde{\Pbf}_{X^nMY^n|F=f}, \quad \forall f\in\{1, \ldots, 2^{nR'}\}. \label{pptilde}
\end{eqnarray}

We now introduce one more random variable that never actually materialized during the implementation.  Let $\tilde{Y}^n$ be the reconstruction sequence intended by the encoder.  The encoding can be considered as a two step process.  First, the encoder selects a $\tilde{Y}^n$ sequence from $\Ccal_f$ with probability proportional to that induced by passing $x^n$ through a memoryless channel given by $\Pbar_{Y|X}$.  Next, the encoder looks up the message index $m(\tilde{Y}^n)$ and transmits it as $M$.

Accordingly, we replace the encoder in the pseudo induced distribution with the two parts discussed:
\begin{equation}
\label{eq:two-step-encoder}
\Pbf_{PPE}(m|x^n,f) = \sum_{\tilde{y}^n} \Pbf_{E1}(\tilde{y}^n|x^n,f) \Pbf_{E2}(m|\tilde{y}^n).
\end{equation}

To analyze the expected distortion performance of the pseudo induced distribution $\tilde{\Pbf}$, we introduce two approximating distributions $\Qbf^{(1)}$ and $\Qbf^{(2)}$.

Let us first define the distribution $\Qbf^{(1)}$:
\begin{eqnarray}
&&\Qbf^{(1)}_{FX^n\tilde{Y}^nMY^n}(f,x^n,\tilde{y}^n,m,y^n)\nonumber\\
&\triangleq&\Pbar_{X^nY^n}(x^n,\tilde{y}^n)\Qbf_{F|\tilde{Y}^n}(f|\tilde{y}^n)\nonumber\\
&&\Pbf_{E2}(m|\tilde{y}^n)\Pbf_{D}(y^n|m,f)\label{spec1}
\end{eqnarray}
where $\Qbf_{F|\tilde{Y}^n}(f|\tilde{y}^n)=\mathbbm{1}\{\tilde{y}^n\in \Ccal_f\}$.  In words, $\Qbf^{(1)}$ is constructed from an i.i.d. distribution according to $\Pbar$ on $(X^n,\tilde{Y}^n)$, two random binnings $F$ and $M$, as specified by the construction of the set of codebooks $\Kcal$, and a decoding of $Y^n$ from the random binnings.

Now we arrive at the reason for using the proportional-probability encoder.  Part 1 of the encoder that selects the $\tilde{Y}^n$ sequences is precisely the conditional probability specified by $\Qbf^{(1)}$:
\begin{eqnarray}
\Qbf^{(1)}_{\tilde{Y}^n|X^nF}(\tilde{y}^n|x^n,f)=\Pbf_{E1}(\tilde{y}^n|x^n,f).
\end{eqnarray}
Therefore, the only difference between the pseudo induced distribution $\tilde{\Pbf}$ and $\Qbf^{(1)}$ is the conditional distribution of $F$ given $X^n$.  This is where Lemma~\ref{conditional-rate} plays a role.

Applying Lemma~\ref{conditional-rate} by identifying $F$ as the uniform binning, since $R'<H_{\Pbar}(Y|X)$, we obtain
\begin{eqnarray}
\Ebb_{\Kcal}\left[\left\Vert\Qbf^{(1)}_{X^nF}-\tilde{P}_{X^nF}\right\Vert_{TV}\right]\leq\epsilon_n^{(rb)}\rightarrow_n 0.
\end{eqnarray}
Using Property \ref{property-tv} $(\ref{d})$, we have
\begin{eqnarray}
\Ebb_{\Kcal}\left[\left\Vert \tilde{\Pbf}_{FX^nY^nM\hat{Y}^n}-\Qbf^{(1)}_{FX^nY^nM\hat{Y}^n}\right\Vert_{TV}\right] \leq \epsilon_n^{(rb)}. \label{PtoQrb}
\end{eqnarray}

The next approximating distribution we define is $\Qbf^{(2)}$:
\begin{eqnarray}
&&\Qbf^{(2)}_{FX^n\tilde{Y}^nMY^n}(f,x^n,\tilde{y}^n,m,y^n)\nonumber\\
&\triangleq& \Qbf^{(1)}_{FX^n\tilde{Y}^nM}(f,x^n,\tilde{y}^n,m)\mathbbm{1}\{y^n=\tilde{y}^n\}.
\end{eqnarray}
Recall from Remark \ref{rem-dec}, decoding $\tilde{Y}^n$ will succeed with high probability if the total rate of the binnings is above the entropy rate of the sequence that was binned. This is well known from the Slepian-Wolf coding result \cite{slepian-wolf1973} \cite{cover-sw1975}. Therefore, since the total binning rate $R+R'>H_{\Pbar}(Y)$, according to the definition of total variation, we obtain
\begin{eqnarray}
\Ebb_{\Kcal}\left[\left\Vert \Qbf^{(1)}_{\tilde{Y}^nY^n}-\Qbf^{(2)}_{\tilde{Y}^nY^n}\right\Vert_{TV}\right]\leq \epsilon_n^{(sw)}\rightarrow_n 0,
\end{eqnarray}
where $\epsilon_n^{(sw)}$ is the decoding error.

Again by Property \ref{property-tv} $(\ref{d})$, we have
\begin{eqnarray}
\Ebb_{\Kcal}\left[\left\Vert \Qbf^{(1)}_{FX^n\tilde{Y}^nMY^n}-\Qbf^{(2)}_{FX^n\tilde{Y}^nMY^n}\right\Vert_{TV}\right] \leq \epsilon_n^{(sw)}. \label{QtoQtrb}
\end{eqnarray}

Combining $(\ref{PtoQrb})$ and $(\ref{QtoQtrb})$ using the triangle inequality, we obtain
\begin{eqnarray}
&&\Ebb_{\Kcal}\left[\left\Vert \tilde{\Pbf}_{FX^n\tilde{Y}^nMY^n}-\Qbf^{(2)}_{FX^n\tilde{Y}^nMY^n}\right\Vert_{TV}\right] \nonumber\\
&\leq& \epsilon_n^{(rb)}+\epsilon_n^{(sw)}. \label{PtoQt}
\end{eqnarray}

Note that the distortion under any realization of $\Qbf^{(2)}$, regardless of the codebook, is
\begin{eqnarray}
\Ebb_{Q^{(2)}}[d(X^n,Y^n)]&=&\Ebb_{Q^{(2)}}[d(X^n,Y^n)]\\
&=&\Ebb_{\Pbar}[d(X,Y)].
\end{eqnarray}

Applying Property \ref{property-tv}$(\ref{b})$, we can obtain
\begin{eqnarray}
&&\Ebb_{\Kcal}\left[\Ebb_{\tilde{\Pbf}}[d(X^n,Y^n)]\right]\nonumber\\
&\leq& \Ebb_{\Pbar}[d(X,Y)]+d_{max}(\epsilon_n^{(rb)}+\epsilon_n^{(sw)}).  \label{rb-bound}
\end{eqnarray}

Furthermore, by symmetry and the law of total expectation, we have
\begin{eqnarray}
&&\Ebb_{\Kcal}\left[\Ebb_{\tilde{\Pbf}}[d(X^n,Y^n)]\right]\nonumber\\
&=&\Ebb_{F}\left[\Ebb_{\Kcal}\left[\Ebb_{\tilde{\Pbf}}[d(X^n,Y^n)] \mid F \right] \right] \\
& = & \Ebb_{\Kcal}\left[\Ebb_{\tilde{\Pbf}}[d(X^n,Y^n)] \mid F=1 \right] \\
& = & \Ebb_{\Kcal}\left[\Ebb_{\Pbf}[d(X^n,Y^n)] \right],
\label{common}
\end{eqnarray}
where the last equality comes from the observation in $(\ref{pptilde})$.

Finally, applying the random coding argument, there exists a code that gives
\begin{eqnarray}
\Ebb_{P}[d(X^n,Y^n)]\leq\Ebb_{\Pbar}[d(X,Y)]+d_{max}\left(\epsilon_n^{(rb)}+\epsilon_n^{(sw)}\right),
\end{eqnarray}
which is less than $D$ for $n$ large enough.

\begin{rem}
This proof method has also proven effective in multi-terminal settings as well. One advantage to this approach is that all auxiliary variables are treated as i.i.d. sequences at some point in the analysis, which is conceptually helpful.
\end{rem}

\begin{rem}
Notice that the error term in the likelihood encoder approach only arises from the soft-covering lemma, while the error terms in the proportional-probability approach come from two places, random binning and MSLC decoding. A non-asymptotic analysis using the proportional-probability approach is given in \cite{eva-thesis}.
\end{rem}

\section{Conclusion} \label{conclude}
In this paper, we have demonstrated how the likelihood encoder can be used to obtain achievability results for various lossy source coding problems. The analysis of the likelihood encoder relies on the soft-covering lemma. Although the proof method is unusual, we hope to have demonstrated that this method of proof is simple, both conceptually and mechanically. The simplicity is accentuated when used for distributed source coding because it bypasses the need for a Markov lemma of any form and it avoids the technical complications that can arise in analyzing the decoder whenever random binning is involved in lossy compression. This proof method applies directly to continuous sources as well with no need for additional arguments, because the soft-covering lemma is not restricted to discrete sources. The likelihood encoder also simplifies analysis in secrecy settings, though this was not demonstrated within this paper. In the secrecy settings \cite{schieler-journal}, \cite{song-sideinfo}, \cite{song-isit2015}, a superposition codebook together with a superposition version of the soft-covering lemma is typically required.

Additionally, a parallel comparison of the achievability technique of \cite{rb-yassaee}, which we dub the ``proportional-probability encoder," has been provided. Our presentation emphasizes the relationship to the likelihood encoder, both operating stochastically with respect to reverse channels of each other.

\section{Acknowledgement} 
This research was supported in part by the National Science Foundation under Grants CCF-1350595, CCF-1116013, CCF-1420575 and ECCS-1343210, and in part by Air Force Office of Scientific Research under Grant FA9550-15-1-0180 and FA9550-12-1-0196.

\ifCLASSOPTIONcaptionsoff
  \newpage
\fi

\bibliographystyle{ieeetr}
\newpage
\bibliography{likelihood_encoder}

\begin{IEEEbiographynophoto}{Eva Song}
(SÕ13, MÕ16) received her master's and PhD degrees in Electrical Engineering from Princeton University in 2012 and 2015, respectively. She received her B.S. degree in Electrical and Computer Engineering from Carnegie Mellon University, Pittsburgh, PA, in 2010. In her PhD work, she studied lossy compression and rate-distortion based information-theoretic secrecy in communications. She is the recipient of Wu Prize for Excellence in 2014. During 2012, she interned at Bell Labs, Alcatel-Lucent, NJ, to study secrecy in optical communications. Her general research interests include: information theory, security, and lossy compression. She joined Amazon in 2015.
\end{IEEEbiographynophoto}

\begin{IEEEbiographynophoto}{Paul Cuff}
(SÕ08, MÕ10) received the B.S. degree in electrical engineering from Brigham Young University, Provo, UT, in 2004 and the M.S. and Ph. D. degrees in electrical engineering from Stanford University in 2006 and 2009. Since 2009 he has been an Assistant Professor of Electrical Engineering at Princeton University.

As a graduate student, Dr. Cuff was awarded the ISIT 2008 Student Paper Award for his work titled ÒCommunication Requirements for Generating Correlated Random VariablesÓ and was a recipient of the National Defense Science and Engineering Graduate Fellowship and the Numerical Technologies Fellowship. As faculty, he received the NSF Career Award in 2014 and the AFOSR Young Investigator Program Award in 2015.
\end{IEEEbiographynophoto}

\begin{IEEEbiographynophoto}{H. Vincent Poor}
(SÕ72, MÕ77, SMÕ82, FÕ87) received the Ph.D. degree in electrical engineering and computer science from Princeton University in 1977.  From 1977 until 1990, he was on the faculty of the University of Illinois at Urbana-Champaign. Since 1990 he has been on the faculty at Princeton, where he is the Dean of Engineering and Applied Science, and the Michael Henry Strater University Professor of Electrical Engineering. He has also held visiting appointments at several other institutions, most recently at Imperial College and Stanford. His research interests are in the areas of information theory, stochastic analysis and statistical signal processing, and their applications in wireless networks and related fields. Among his publications in these areas is the recent book Mechanisms and Games for Dynamic Spectrum Allocation (Cambridge University Press, 2014).

Dr. Poor is a member of the National Academy of Engineering and the National Academy of Sciences, and is a foreign member of Academia Europaea and the Royal Society. He is also a fellow of the American Academy of Arts and Sciences, the Royal Academy of Engineering (U. K.), and the Royal Society of Edinburgh.  In 1990, he served as President of the IEEE Information Theory Society, in 2004-07 as the Editor-in-Chief of these TRANSACTIONS, and in 2009 as General Co-chair of the IEEE International Symposium on Information Theory, held in Seoul, South Korea. He received a Guggenheim Fellowship in 2002 and the IEEE Education Medal in 2005. Recent recognition of his work includes the 2014 URSI Booker Gold Medal, the 2015 EURASIP Athanasios Papoulis Award, the 2016 John Fritz Medal, and honorary doctorates from Aalborg University, Aalto University, HKUST, and the University of Edinburgh.
\end{IEEEbiographynophoto}

\end{document}